\pdfoutput=1
\RequirePackage{boolean}
\RequirePackage{reversion}

\IfFileExists
    {\jobname.cfg}
    {\input{\jobname.cfg}}
    {\input{article.cfg}}

\documentclass[\Anonymous{anonymous,review,}{}\ACM{}{nonacm,}acmsmall,10pt,screen]{acmart}
\authorsaddresses{}

\ACM{}{
  \usepackage{hyperref}
  \hypersetup{colorlinks=true, linkcolor=Red,%
    citecolor=ForestGreen, urlcolor=RoyalBlue}
}

\usepackage[utf8]{inputenc}
\usepackage[T1]{fontenc}
\usepackage[english]{babel}
\usepackage{multicol}

\usepackage{XXX}
\newXXX{\XGS}{XGS}{blue}
\newXXX{\XPE}{XPE}{magenta}
\newXXX{\XLR}{XLR}{orange}

\newXXX{\XAG}{XAG}{brown}

\usepackage{mathpartir}
\usepackage{amsmath}
\usepackage{amssymb}

\usepackage{normalization_by_realizability}
\usepackage{stlc}

\usepackage{xspace}
\newcommand{\ie}{\textit{i.e.}\xspace}
\newcommand{\etc}{\textit{etc.}\xspace}

\makeindex

\setcopyright{none}             

\let\para\paragraph
\renewcommand{\paragraph}[1]{\para{\textbf{#1}}}

\begin{document}
\title{Functional Pearl: a Classical Realizability Program}

\author{Pierre-{\'E}variste Dagand}
\affiliation{
  \institution{Sorbonne Université, CNRS, Inria Paris, LIP6}
  \country{France}
}
\email{pierre-evariste.dagand@lip6.fr}

\author{Lionel Rieg}
\affiliation{
  \institution{VERIMAG UMR 5104, Université Grenoble Alpes, France}
  \country{France}
}
\email{lionel.rieg@univ-grenoble-alpes.fr}

\author{Gabriel Scherer}
\affiliation{
  \institution{Inria Saclay}
  \country{France}
}
\email{gabriel.scherer@inria.fr}

\begin{abstract}
  For those of us who generally live in the world of syntax, semantic
  proof techniques such as reducibility, realizability or logical
  relations seem somewhat magical despite -- or perhaps due to --
  their seemingly unreasonable effectiveness. Why do they work?  At
  which point in the proof is ``the real work'' done?
  Hoping to build a programming intuition of these proofs, we
  \emph{implement} a normalization argument for the simply-typed
  $\lambda$-calculus with sums: instead of a proof, it is described as
  a program in a dependently-typed meta-language.

  The semantic technique we set out to study is Krivine's
  \emph{classical realizability}, which amounts to
  a proof-relevant presentation of \emph{reducibility} arguments --
  unary logical relations. Reducibility assigns a predicate to each
  type, realizability assigns a set of \emph{realizers}, which are
  abstract machines that extend $\lambda$-terms with a first-class
  notion of contexts. Normalization is a direct consequence of an
  \emph{adequacy theorem} or \emph{fundamental lemma}, which states that
  any well-typed term translates to a realizer of its type.
  We show that the adequacy theorem, when written as a dependent
  program, corresponds to an evaluation procedure. In particular,
  a weak normalization proof precisely computes a series of reductions
  from the input term to a normal form. Interestingly, the choices
  that we make when we define the reducibility predicates -- truth and
  falsity witnesses for each connective -- determine the evaluation
  order of the proof, with each datatype constructor behaving in
  a lazy or strict fashion.

  While most of the ideas in this presentation are folklore among
  specialists, our dependently-typed functional program  provides an
  accessible presentation to a wider audience. In particular, our work
  provides a (hopefully) gentle introduction to abstract machine calculi which
  have recently been used as an effective research
  vehicle~\citep*{downen-icfp15,curien-fiore-munch-popl16}.
\end{abstract}


\maketitle

\XXX{Comments are included in \texttt{paper.pdf} and omitted in \texttt{icfp.pdf},
  which also omits the content that is marked for the Long version only.}

\begin{version}{\Long}
\end{version}

\Long{
}{
}

\section{Introduction}

Realizability, logical relations and parametricity are tools to study
the meta-theory of syntactic notions of computation; typically, typed
$\lambda$-calculi. Starting from a syntactic type system, they assign to
each type/formula a predicate that captures a semantic property, such
as normalization, consistency, or some canonicity properties. Proofs
using such techniques rely crucially on an \emph{adequacy lemma}, or
\emph{fundamental lemma}, that asserts that any term accepted by the
syntactic system of inference also verifies the semantic property
corresponding to its type. In a realizability setting, one would prove
that the term is accepted by the predicate. Similarly, the fundamental
lemma of logical relations states that the term is logically related
to itself.  There are many variants of these techniques, which have
scaled to powerful logics~\citep*{abel-hdr} (e.g., the
Calculus of Constructions, or predicative type theories with
a countable universe hierarchy) and advanced type-system
features~\citep*{turon-thamsborg-ahmed-birkedal-dreyer,timany-2018}
(second-order polymorphism, general references, equi-recursive
types, local state\dots).

This paper sets out to explain the \emph{computational
  behavior} of such a semantic approach.
Where should we start looking if we are interested in all three of
realizability, logical relations and parametricity? Our reasoning was
the following.
First, we ruled out parametricity: it is arguably a specific form of
logical
relation~\citep*{hermida:logical-relation-parametricity-14}. Besides,
it is motivated by applications -- such as characterizing
polymorphism, or extracting invariants from specific types -- whose
computational interpretation is less clear than a proof of
normalization.
Second, general logical relations seem harder to work with than
realizability. They are binary (or $n$-ary) while realizability is
unary, and the previous work of \citet*{bernardy-lasson} suggests that
logical relations and parametricity can be built from realizability by
an iterative process. While the binary aspect of logical relations is
essential to formulate representation invariance
theorems~\citep*{atkey:parametric-dependent-types-14}, it does not come
into play for normalization proof of simpler languages such as the
simply-typed $\lambda$-calculus.
We are therefore left with realizability while keeping an eye towards
normalization proofs -- weak normalization rather than strong, for the
benefit of simplicity. This is a well-worn path, starting with the
reducibility method~\citep*{tait1967}, which we shall now recall and
put in perspective with recent developments.
While the present work focuses on the pedagogical value of revisiting
the Classics~\citep{krivine:lambda-models}, the scope of realizability
techniques extends beyond meta-theoretical considerations, with
applications to compiler correctness~\citep{benton:biorth-compiler}
or language design~\citep{downen:sequent-core} to name but a few
examples.

The notion of \emph{polarity} plays a central role in the following
retrospective. We split our types into \emph{positives} and
\emph{negatives}. This distinction has a clear logical status: it
refers to which of the left or right introduction rules, in a sequent
presentation, are
non-invertible~\citep*{andreoli:focusing-92,zeilberger:harmony-focusing-2013}.
This distinction is also reminiscent of
call-by-push-value~\citep*{levy:cbpv} and can be summarized by the
adage: ``positive types are defined by their constructors whereas
negatives types are defined by their destructors''.  The sum type,
integers and inductive datatypes in general are positives whereas
functions, records and coinductive types~\citep*{abel:copatterns} in
general are negatives. ML programmers tend to favor the positive
fragment to architect their software whereas object-oriented programmers live almost exclusively
in the negative fragment~\citep*{cook:data-abstraction}.
In the absence of side-effects, we may sometimes playfully encode the
former using the latter, through Church encodings for example, or we
may gainfully defunctionalize the former to obtain the
latter~\citep{danvy:program-as-data-objects}. However, as noted
before~\citep{zeilberger-phd} and as we see in the following,
studying typed calculi through suitably polarized lenses yields
a refined understanding of past research results.

\paragraph{On the negative fragment: reducibility}

Normalization proofs for typed $\lambda$-calculi cannot proceed by
direct induction on the term structure: to prove that a function
application $\app \ea \eb$ is normalizing, a direct induction would
only provide as hypotheses that the function $\ea$ and its argument
$\eb$ are normalizing. But this tells us nothing of whether
$\app \ea \eb$ itself normalizes; for example, the self-application
$\ob\omega \eqdef \lam \eva {\app \eva \eva}$ is normalizing, but applying
it to itself gives a non-normalizing term $\app {\ob\omega} {\ob\omega}$
that reduces to itself.

A powerful fix to this issue, introduced by \citet*{tait1967}, is to
define a reducibility predicate over terms by induction over
types:
(1)
Terms $\ea$ of base type $\ta$ are reducible at this type if they are
strongly normalizing;
(2)
Terms $\ea$ of function type
$\fun \ta \tb$ are reducible at this type if for all $\eb$ reducible
at $\ta$, the application $\app \ea \eb$ is reducible at $\tb$.
%
Reducibility is thus defined to imply normalization, but also provide
stronger hypotheses to an inductive argument; in particular, reducibility of
function types is exactly what we needed in the $\app \ea \eb$
case. One can then prove, by induction on typing derivations, that any
well-typed term $\ea$ of type $\ta$ is reducible (and thus
normalizing) at $\ta$.

There is more to this technique than a mere proof trick of
strengthening an inductive hypothesis. In particular, we have
transformed a \emph{global} property of being normalizing into
a \emph{modular} property, reducibility, which describes the interface
expected of terms at a given type. A global, whole-term property does
not say much of what happens when a term is put in a wider context;
the modular property specifies the possible interactions between the
term and its context.

This view directly relates to more advanced fields of logic and type
theory. It offers a semantic interpretation of types, where ``behaving
at the type $\ta$'' is a property of how a (possibly untyped) term
interacts with outside contexts; syntactic typing
  rules can be reconstructed as admissible proof
principles~\citep*{harper-semantic-types} that establish this
behavioral property by syntactic inspection of the term. In logic,
realizability is a general notion of building models where provable
formulas are \emph{realized} by computable functions, and it is through
this lens that the constructive character of intuitionistic logics was
historically explained. In programming language theory, it is common
to define not only a unary predicate in such a type-directed fashion,
but also binary predicates, called logical relations, to capture the
notion of program equivalence.

\paragraph{Negative interlude: a simpler proof of weak normalization}

Note that while reducibility is traditionally used to prove strong
normalization -- the fact that all possible reduction sequences (for all
reduction strategies) are terminating -- it can also be used to prove
weak normalization -- the fact that there exists one terminating
reduction sequence. For the negative fragment of the simply-typed
$\lambda$-calculus, there is a simpler proof argument for weak
normalization, attributed to Alan Turing~\citep*{barendregt-manzonetto-2013}.
This argument does not proceed by induction on the typing
derivation, but considers the sizes of the types appearing in redexes,
and reduces one of the redexes with the largest types.
The intuition is that when a redex $\app {(\lam \eva \ea)} \eb$
at type $\fun \ta \tb$ is reduced to $\subst \ea \eva \eb$,
it may create new redexes, but only at the smaller types $\ta$ or $\tb$.

However, introducing sum types $\sum \ta \tb$ somewhat clouds this
picture. The sum elimination construction is typed as follows
\begin{mathpar}
  \infer
  {\ca \der \ty \ea {\sumfam \ta}\\
   \ca, \ty {\evaidx 1} {\taidx 1} \der \ty {\ebidx 1} \tc\\
   \ca, \ty {\evaidx 2} {\taidx 2} \der \ty {\ebidx 2} \tc}
  {\ca \der \ty {\desumfam \ea \eva \eb} \tc}
\end{mathpar}
where $\tc$ is unrelated in size to either $\taidx 1$ or $\taidx 2$.  For
example, consider the following term:
$$
\desum {\left(\desum {\inj{1}{\enfa}}
                     {\evaidx 1} {\inj{2}{(\inj{2}{\evaidx 1})}}
                     {\evaidx 2} {\inj{2}{(\inj{1}{\evaidx 2})}}\right)}
       {\evaidx 1} {0}
       {\evaidx 2} {1}
$$
where $\enfa$ is a value of type $\taidx 1$. The (only) available redex
is the innermost case, at type $\sumfam \ta$. However, performing the
reduction yields the term
$$
\desum {\inj{2}{(\inj{2}{\enfa})}}
       {\evaidx 1} {0}
       {\evaidx 2} {1}
$$
that introduces a redex of (larger!) type
$\sum{\ta_3}{(\sum{\taidx 2}{\taidx 1})}$. More generally, when the term
$\ea$ in the inference rule above is a sum injection $\inj i \eap$;
reducing this redex at type $\sumfam \ta$ into $\subst {\ebidx i}
{\evaidx i} \eap$ may create redexes at type $\taidx i$, but also at type
$\tc$, which is unrelated in size. This typically occurs when
$\ebidx i$ starts with a constructor of type $\tc$ and the whole
redex is under an elimination form for $\tc$, as exemplified by the
above nested cases.

One can work around this issue by introducing commuting
conversions. For instance, the nested cases above can
be extruded as follows
$$
\desum {\inj{1}{\enfa}}
       {\evaidx 1} {\desum {\inj{2}{(\inj{2}{\evaidx 1})}}
                        {\evaidx 1} {0}
                        {\evaidx 2} {1}\medskip}
       {\evaidx 2} {\desum {\inj{2}{(\inj{1}{\evaidx 2})}}
                        {\evaidx 1} {0}
                        {\evaidx 2} {1}}
$$
whose effect is to enable the redexes of larger type -- see for
example \citet*[chapter 3.3]{scherer:PhD} for a systematic treatment of commuting conversions in a weakly normalizing setting. However, this approach adds bureaucracy
and complexity.

Trying to naively extend the reducibility method to sum types hits a
similar well-foundedness issue. A natural definition of reducibility
at a sum type $\sumfam \ta$ would be that $\ea$ is reducible at
$\sumfam \ta$ if, for any type $\tc$ and $\ebidx 1, \ebidx 2$
reducible at $\tc$, the elimination $\pddesumfam \ea \eva \eb$ is
reducible at $\tc$. However, this definition is not well-founded. In
the function case, we defined reducibility at $\fun \ta \tb$ from
reducibility at the smaller types $\ta$ and $\tb$, but to define
reducibility at type $\sumfam \ta$ we use the definition of reducibility
at all types $\tc$, in particular at $\sumfam \ta$ itself or even
larger types, such as $\sum{\taidx 3}{(\sum{\taidx 2}{\taidx 1})}$
in our example above.

\paragraph{On the positive fragment: bi-orthogonal closure}

To obtain a reducibility proof accommodating sum types and side-effects,
\citet*{lindley-stark-2005} solved this well-foundedness problem. They use a
bi-orthogonality closure, also called the top-top ($\top\top$)
closure~\citep*{pitts-stark-1998} following the notation
of \citet*{pitts:biorthogonality-00}, although ($\bot\bot$) would be a
much more appropriate name and notation.  Looking through the
Curry-Howard isomorphism, we find that (linear) logicians had in fact
pioneered this technique: \citet{girard:linear-logic} introduced
bi-orthogonality to prove normalization of linear logic proof-nets.
We follow in these footsteps (Section~\ref{sec:types}).

This technique proceeds as follows. We say that a program
\emph{context} $\plug \Ca \ehole$ is \emph{reducible at type $\ta$} if,
for any \emph{value} $\va$ of type $\ta$, the composed term
$\plug \Ca \va$ is normalizing. Then we define the terms reducible at
a sum type $\sum \ta \tb$ to be the $\ea$ such that, for any context
$\plug \Ca \ehole$ reducible at type $\sum \ta \tb$, the composed term
$\plug \Ca \ea$ is normalizing.

The guiding intuition of our first attempt at a definition, and of the
definition of reducibility at function types, was to express how
using the term $\ty \ea \ta$ preserves normalization. This
intuition still applies here; in particular, the set of contexts
$\plug \Ca \ehole$ includes the contexts of the form
$\ddesumfam \ehole \eva \eb$, that perform a case analysis on their
argument. However, notice that our definition does not require that
$\plug \Ca \ea$ be reducible at its type, which would again result in
an ill-founded definition, but only that it be normalizing: we do not
require the modular property, only the global one.

Note that adapting the naive definition by saying that $\ea$ is
reducible at $\sumfam \ta$ if $$\ddesumlam \ea {\evaidx \idx} {\ebidx \idx}$$ is normalizing
would give a well-founded definition, but one that is too weak, as the
global property is not modular. For example, we would not be able to
prove that a term of the form $\app {(\matchwith \ea \dots)} \eap$ is
normalizing by induction on its typing derivation. This is why we need
to quantify on all contexts $\plug \Ca \ehole$, such as
$\app {(\matchwith \ehole \dots)} \eap$ in our example, rather than
only on case eliminations.

There are really three distinct classes of objects in such
a reducibility proof. The \emph{term} and \emph{context fragments}
interact with each other, respecting a modular interface given by the
reducibility predicate at their type. \emph{Whole programs} are only
required to be normalizing. These whole programs are formed by the
interaction of a term fragment $\ea$ and a context $\plug \Ca \ehole$, both
reducible at a type $\ta$.

Note that it is possible to adapt the standard reducibility arguments
in presence of sums ($\ea$ is reducible at $\sumfam \ta$ if it reduces
to $\inj{i} {\eaidx i}$ for $\eaidx i$ reducible at $\taidx i$) but
this requires strengthening various (global) definitions to enforce
stability under anti-reduction, which is a local property stating that
if a predicate holds for $\plug \Ca {\subst {\ebidx i}{\evaidx
i}{\eaidx i}}$ then it must hold for $ {\plug \Ca {\ddesumfam
{{\inj{i}{\eaidx i}}} \eva \eb}}$ as well. This approach is thus less
modular.

\medskip

With this paper, we wish to turn these nuggets of wisdom into a
rationalized process. We achieve this by expressing the problem in a setting
familiar to functional programmers -- dependent types -- and by
expressing the problem in a conceptual framework  (Section~\ref{sec:background}) so
powerful as to make its solution simple and illuminating. This paper
proposes the following 4-step program to enlightenment:
\begin{itemize}
\item
  We recall the standard framework of classical realizability for the
  simply-typed $\lambda$-calculus with sums while identifying the
  precise role of (co)-terms and (co)-values
  (Section~\ref{sec:background}).
  This distinction is absent from usual presentations of classical
  realizability~\citep*{miquel:classical-rea,rieg:PhD} because of
  their bias toward negative types, whereas it plays a crucial role in
  our study of normalization proofs;
\item We show that the computational content of the adequacy lemma is an evaluation
  function (Section~\ref{sec:nbr-variant1}). Along the way, we
  re-discover the fact that the polarity of object types determines
  the flow of values in the interpreter;
\item We show that this evaluation order can be exposed by \emph{recovering} the
  compilation to abstract machines from the typing constraints that
  appear when we move from a simply-typed version of the adequacy
  lemma to a dependently-typed version capturing the full content of
  the theorem (Section~\ref{subsec:dependent});
%
\item There are several possible ways to construct the truth and
  falsity witnesses used in the proof. We show that these choices
  mechanically sets the evaluation order of the resulting
  normalization procedure (Section~\ref{sec:nbr-all}). At the
  computational heart of a normalization proof lies an evaluator whose
  evaluation strategy is dictated by the flow of values in the model.
\end{itemize}

This research program is also a constructive one: throughout this
paper, we demonstrate that, with some care, our argument%
\footnote{Available as supplementary material.}
can be carried in the Coq proof assistant~\citep{coq} and that the
simply-typed normalization function from
Section~\ref{sec:nbr-variant1} can be extracted from the
dependently-typed program of Section~\ref{subsec:dependent}. We hope
for this paper to excite logicians -- by making fascinating
techniques palatable to a larger audience of functional programmers --
while feeding functional programmers' thirst for Curry-Howard
phenomena.



\section{Background}
\label{sec:background}

\begin{figure}[htb]
\begin{mathpar}
  \begin{array}{rcl@{\quad}r}
    \ea,\eb & := &         & \text{terms} \\
        & \mid & \eva, \evb, \evc   & \text{variables} \\
        & \mid & \lam \eva \ea      & \text{abstractions} \\
        & \mid & \app \ea \eb       & \text{applications} \\
        & \mid & \inj 1 \ea \mid \inj 2 \ea & \text{injections} \\
        & \mid & \ddesumfam \ea \eva \eb & \text{case analysis} \\
        & \mid & \ldots & \\
  \end{array}

  \begin{array}{rcl@{\quad}r}
    \ta,\tb \in \type & := & \tpa  & \text{positive types} \\
            & \mid & \tna & \text{negative types}
    \end{array}\\

  \begin{array}{rcl@{\quad}r@{\qquad}rcl@{\quad}r}
    \tpa & := & \tnat  & \text{integers}
    &
    \tna & := & \tunit & \text{unit type} \\
         & \mid & \sum \ta \tb & \text{sum type}
    &
         & \mid & \fun \ta \tb & \text{function type} \\
  \end{array}

\begin{array}{c@{\qquad}c}
\multicolumn{2}{c}{
  \obinfer{(\eva \of \ta) \in \ca}
  {\judgty \ca \eva \ta}}
  \qquad\qquad
\bigskip\\
  \obinfer
  {\judgty {\ca, \eva \of \ta} \ea \tb}
  {\judgty \ca {\lam \eva \ea} {\fun \ta \tb}}
&
  \obinfer
  {\judgty \ca \ea {\fun \ta \tb}\\
   \judgty \ca \eb \ta}
  {\judgty \ca {\app \ea \eb} \tb}
\bigskip\\
  \obinfer
  {\judgty \ca \ea {\taidx i}}
  {\judgty \ca {\inj i \ea} {\sumfam \ta}}
&
  \obinfer
  {\judgty \ca \ea {\sumfam \ta}\\
   \judgty {\ca, \evaidx i \of \taidx i} {\ebidx i} \tc}
  {\judgty \ca {\ddesumfam \ea \eva \eb} \tc}
\end{array}

\end{mathpar}
\caption{Source language: simply-typed $\lambda$-calculus}
\label{fig:lambda}
\end{figure}

This section introduces various notions this work relies on:
realizability, with a focus on classical realizability, and the
$\mumutilda$ abstract machine.
%
%
\figurename~\ref{fig:lambda} defines our object of study: the simply-typed
$\lambda$-calculus, with functions and sums, integers and a unit type. We dispense with the (folklore) reduction rules as well as term formers and typing rules for unit and integer types. Additionally, we define the Booleans as $\tbool \triangleq \ob{\sum \tunit \tunit}$, inhabited by $\etrue$ (first injection) and $\efalse$ (second injection). We focus on sums and functions throughout the paper, the remaining cases being expounded in our Coq development: here, their role is limited to providing concrete base cases for type-directed definitions.
As we mentioned in the introduction, we distinguish positive
and negative types: our arrows and unit type are negative whereas our sums and
integers are positive.
This distinction, which has a clear logical motivation, matters
when defining certain objects in our proofs.

\label{sec:colors}
\paragraph{Notation} We color syntactic objects (and the meta-language variables used to
denote them) in \obcolor{} to make it easier to distinguish
object-level constructs from meta-level constructs.  This becomes
particularly useful after Section~\ref{sec:nbr-variant1},
where we reify the meta-language into a program.
We never mix
identifiers differing only by their color -- you lose little if you do
not see the colors.

\subsection{Realizability interpretations}

In mathematical logic, the name \emph{realizability} describes a
family of computational interpretations of logic, introduced by
Kleene in 1945 to understand intuitionistic arithmetic.
Consider a logic with formulas $\ta$, and a set of syntactic inference
rules for judgments of the usual form $\ca \mathrel{\ob{\der}} \ta$. We may wonder
whether this logic is sound -- cannot prove contradictions -- and
perhaps hope that it is constructive in some suitable sense.
A realizability interpretation is given by a choice of syntax of
\emph{programs} $\ra$, and a set of \emph{realizers} $\tru \ta$ for each
formula $\ta$, such that (1) false formulas have no realizers and
(2) provable closed formulas $\mathrel{\ob{\der}} \ta$ have a realizer
$\ra \in \tru \ta$. This yields a mathematical model that justifies
the inference rules, and provides a computational interpretation
of the logic.

The realizability relation $\ra \in \tru \ta$ -- often written
$\ra \realizes \ta$ -- plays a very different role from a typing
judgment $\judgty{} \ra \ta$.
A type system ought to be modular and checking types ought to be decidable.
Realizability can be defined in arbitrary ways, and tends to be wildly undecidable.
Indeed, typing is a structural property describing how a program is
built whereas realizability is only concerned with reduction,
regardless of what a realizer actually is. For example, the program
$\ifte{\efalse}{\etrue}{42}$ is clearly ill-typed whereas, from a
reduction standpoint, it amounts to the integer $\ob{42}$ and shall
therefore be a realizer of the type $\tnat$.

In computer science, realizability interpretations have been used to
build models of powerful dependently-typed logics with set-theoretic
connectives in the PRL family of proof assistants, or the soundness of
advanced type systems for languages with
side-effects~\citep{brunel-phd,lepigre:pml-2016}.
\emph{Classical realizability}, presented next, is a sub-family of
realizability interpretations whose realizers are abstract
machines. It was introduced by Krivine in the 90s to study the
excluded-middle, using control operators to realize it and thus
providing (constructive) models of classical
logic~\citep{krivine:realizability-classical-logic} and
Zermelo-Fraenkel set theory~\citep{krivine:ZF}.

\subsection{The Krivine machine: classical realizability in the negative fragment}
\label{sub:rea-negative}

As argued in the introduction, reasoning about $\lambda$-terms leads
to considering not only the \emph{terms} but also the \emph{contexts}
(with which the interaction occurs) and \emph{whole programs} (which
are just asked to normalize). Abstract machines, such as the Krivine
Abstract Machine~\citep*{krivine:machine}\footnote{Although the
  reference paper is dated from 2007, as its author said, the KAM “was
  introduced twenty-five years ago”. Its first appearance was in an
  unpublished but widely circulated note \citep*{krivine:KAM-85}
  available on the author’s web page.}
given in \figurename~\ref{fig:KAM}, turn this conceptual distinction into
explicit syntax: it includes \textbf{t}erms $\ea$ but also an explicit
syntactic category of \textbf{e}valuation contexts $\ka$ as stacks of arguments (context formers are
$\cons \ea \ka$, which pushes an applied argument $\ea$ into the
context $\ka$, and an end-of-stack symbol $\kstar$) as well as
\emph{machine configurations} $\mach \ea \ka$,
pairing a term and a context together to compute.

\begin{figure}
  \begin{minipage}[t]{.4\linewidth}
    \hspace{5em}\textbf{Syntax}
    \begin{mathpar}
      \begin{array}{l@{\quad\gramdef\quad}lr}
        \ea & \eva \mid \lam \eva \ea \mid \app \ea \eb & terms \\
        \ka & \kstar \mid \cons \ea \ka                 & contexts \\
        \ma & \mach \ea \ka                             & machines \\
      \end{array}
    \end{mathpar}
  \end{minipage}
  \hfill
  \begin{minipage}[t]{.56\linewidth}
    \hspace{5em}\textbf{Machine Reduction}
    \begin{align}
      \mach {\app \ea \eb} \ka
      & \rto \mach \ea {\cons \eb \ka}
      \label{red:app}
      \\
      \mach {\lam \eva \ea} {\cons \eb \ka}
      & \rto \mach {\subst \ea \eva \eb} \ka
      \label{red:lam}
    \end{align}
  \end{minipage}
  \caption{The Krivine Abstract Machine}
  \label{fig:KAM}
\end{figure}

\paragraph{Remark} For this simplified abstract machine, the terms
$\ea$ of the machine configuration are exactly the usual (untyped)
$\lambda$-terms -- with only function types, no sums, integers, or unit type.
This is not the case in general, so we are careful to distinguish
the two syntactic categories; given a closed $\lambda$-term $\ea$, we write
$\transl \ea$ for its (identity) embedding as a machine term.
\begin{version}{\Long}
Traditionally, logicians define type systems on top of realizers as
a way to systematically realize large classes of formulas. In these
systems, the terms \emph{are} the realizers. However, recent work in
classical realizability~\citep{miquel:implicative-algebras}
demonstrates the value of decoupling both worlds, whereby one
effectively compiles terms to realizers and does not even need
a syntax for realizers, which are seen as black boxes.
\end{version}

Figure~\ref{fig:KAM} also defines the reductions of this machine.
Function application in a configuration $\mach {\app \ea \eb} \ka$ merely
pushes its argument in the context $\mach \ea {\cons \eb \ka}$.
Such application contexts get reduced when they meet
a $\lambda$-expression:
$\mach {\lam \eva \ea} {\cons \eb \ka}$ reduces to
$\mach {\subst \ea \eva \eb} \ka$.
This abstract machine simulates weak call-by-name reduction in the
$\lambda$-calculus: the $\lambda$-term $\ea$ reduces to $\eb$ in the
weak call-by-name strategy if and only if the configuration
$\mach {\transl \ea} \kstar$ reduces to $\mach {\transl \eb} \kstar$.
(In particular, if $\mach {\transl \eb} \kstar$ does not reduce, then
$\eb$ it is a weak normal-form of $\ea$.)
A machine configuration is a self-contained object, which has no other
interaction with the outside world. Our realizability interpretation
is parametrized over the choice of a \emph{pole} $\Pole$, a set
of configurations exhibiting a specific, good or bad, behavior.

\begin{definition}
  A \emph{realizability structure}~\citep{krivine:realizability-structure} is a triple
  $(\mathbb{T}, \mathbb{E}, \Pole)$ where $\mathbb{T}$ is a set of
  machine terms, $\mathbb{E}$ a set of machine contexts,
  and $\Pole$ is a set of configurations $\mach \ea \ka$ with
  $\ea \in \mathbb{T}$ and $\ka \in \mathbb{E}$ such that:
  \begin{itemize}
  \item $\mathbb{T}, \mathbb{E}$ are closed by context-formers:
    \begin{mathpar}
      \kstar \in \mathbb{E}

      \ea \in \mathbb{T} \ \wedge\ \ka \in \mathbb{E} \implies (\cons \ea \ka) \in \mathbb{E}
    \end{mathpar}
  \item $\Pole$ is closed by anti-reduction (reduction is given by rules~(\ref{red:app}) and~(\ref{red:lam}) in Figure~\ref{fig:KAM}):
    \begin{equation}\label{eqn:anti-red}
      \mach {\ea'} {\ka'} \in \Pole
      \ \wedge\ %
      \mach \ea \ka \rto \mach {\ea'} {\ka'}
      \implies
      \mach \ea \ka \in \Pole
      \end{equation}
  \end{itemize}
\end{definition}
For instance, to capture ``configurations that reduce to a normal form'', we would use the realizability structure where $\Terms$ is the set of closed terms, $\Coterms$ is the set of closed contexts, and $\Pole$ is $\left\{ \ma \,\middle|\, \exists\, \map\!, \ma \rtos \map \land \lnot(\exists\, \mapp\!, \map \rto \mapp)\right\}$.
Another useful definition of the pole is, for example, ``being weakly
normalizing with $\ob{\injName 1}$ or $\ob{\injName 2}$ as the head constructor of the normal form'', formally $\Pole \eqdef \left\{ \ma \,\middle|\, \exists\, \map\!, \ma \rtos \map \land \lnot(\exists\, \mapp\!, \map \rto \mapp) \land (\exists \ea \, \ka, \map = \mach{\inj 1 \ea} \ka \lor \map = \mach{\inj 2 \ea} \ka) \right\}$.
Instantiating Theorem~\ref{thm:fundamental-thm} with this pole lets us deduce that closed terms of type $\sum \ta \tb$ have a canonical form.

The intuition behind the anti-reduction closure is that we are interested in
behavioral properties of machines: normalizing, reducing to
an integer or a well-typed value, performing some side effect (assuming that the machine allows it), \etc
Anti-reduction reflects the fact that such properties may require a few reduction steps to be fulfilled: a machine \emph{eventually}
normalize (to an integer or to a well-typed value),
\emph{eventually} perform some side effect, \etc

\paragraph{Orthogonality}
Configurations are self-contained, whereas a term needs a context to compute,
and vice-versa. For $\mathcal{T} \subseteq \Terms$ an arbitrary set of
terms, we define its \emph{orthogonal} $\orth{\mathcal{T}}$ as the set of contexts that
compute well against terms in $\mathcal{T}$, in the sense that
they end up in the pole -- they exhibit the property we are interested in.
The orthogonal of a set of contexts
$\mathcal{E} \subseteq \Coterms$ is defined by symmetry:

\noindent
\begin{minipage}{.49\linewidth}
\begin{equation} \label{eqn:orthT-Set}
  \orth{\mathcal{T}} \eqdef \{ \ka \in \mathbb{E}
    \mid \forall \ea \in \mathcal{T}, \mach{\ea}{\ka} \in \Pole \}
\end{equation}
\end{minipage}
\hfill
\begin{minipage}{.49\linewidth}
\begin{equation}  \label{eqn:orthF-Set}
  \orth{\mathcal{E}} \eqdef \{ \ea \in \mathbb{T}
    \mid \forall \ka \in \mathcal{E}, \mach{\ea}{\ka} \in \Pole \}
\end{equation}
\end{minipage}
\smallskip

Orthogonality behaves in a way that is similar to intuitionistic
negation: for any set $S$ of either terms or contexts we have
$S \subseteq \biorth{S}$ (just as $P \implies \neg{\neg{P}}$) and this
inclusion is strict in general. If one understands the pole as
defining a notion of \emph{observation}, the bi-orthogonals of a given
(co)-term are those that have the same observable behavior. In
particular, we think of sets $S$ that are \emph{stable by
bi-orthogonality} ($S$ is isomorphic to $\biorth{S}$) as observable
predicates on terms or contexts: they define an extensional property
that is not finer-grained than what we can observe through the pole.
In particular, sets $\orth{S}$ that are themselves orthogonals are stable by
bi-orthogonality: $\triorth{S} \subseteq \orth{S}$ (just as
$\neg\neg\neg{P} \implies \neg{P}$), so $\triorth{S} \eqset \orth{S}$, where we write $S \eqset T$ to say that the sets $S$
and $T$ are in bijection and reserve the equality symbol to the
definitional equality.

\begin{example}
For any pole $\Pole$, the singleton set \(S \eqdef \{ \lam{\eva}{\eva} \}\) is
not stable by bi-orthogonality as \(\app{(\lamtwo \evb \eva \eva)}{\omega}\)
does not belong to~$S$ whereas it belongs to~$\biorth S$.
Indeed, we have
\[
\begin{array}{r@{\iff}l}
  \app{(\lamtwo \evb \eva \eva)}{\omega} \in \biorth{\{ \lam \eva \eva \}}
  & \forall \ka \in \orth{\{ \lam \eva \eva \}},
    \mach{\app{(\lamtwo \evb \eva \eva)}{\omega}} \ka \in \Pole \\
  & \forall \ka,\; \mach{\lam \eva \eva} \ka \in \Pole \to
    \mach{\app{(\lamtwo \evb \eva \eva)}{\omega}} \ka \in \Pole \\
\end{array}
\]
By anti-reduction~(\ref{eqn:anti-red}),
\(
  \mach{\app{(\lamtwo \evb \eva \eva)}{\omega}} \ka
  \rto^{(\ref{red:app})} \mach{\lamtwo \evb \eva \eva}{\cons \omega \ka}
  \rto^{(\ref{red:lam})} \mach{\lam \eva \eva} \ka  \in \Pole
\),
we therefore conclude that
$\mach{\app{(\lamtwo \evb \eva \eva)}{\omega}} \ka \in \Pole$.
Intuitively, the closure $\biorth S$ is the set of terms that cannot
be distinguished from the term $\lam \eva \eva$ by picking a context
and observing whether the resulting machine is in the pole. Here, it
allows us to encompass all the functions that behave like the identity
function from the standpoint of the observation defined by the
pole.

\end{example}

We can now give a (classical) realizability interpretation of the
types $\ta, \tb, \dots$ of the simply-typed $\lambda$-calculus (again, no sum, integers nor unit),
parametrized over a choice of realizability structure
$(\mathbb{T}, \mathbb{E}, \Pole)$. Rather than just defining sets of
realizers $\tru \ta$, we also define define a set of contexts
$\fal \ta$ called the \emph{falsity witnesses} of $\ta$.
For consistency, the set of terms $\tru \ta$ (the realizers of $\ta$) is called the \emph{truth witnesses} of $\ta$.
They are called this way because, from a logical point of view, $\tru \ta$
contains \emph{justifications} for $\ta$, while $\fal \ta$ contains
\emph{refutations} for $\ta$. Their definitions imply that a term in
$\tru \ta$ and a context in $\fal \ta$ can interact
according to the interface $\ta$.
\begin{align}
  \falv {\fun \ta \tb}
  & \eqdef \{ \cons \eb \ka \in \Coterms \mid \eb \in \tru \ta, \ka \in \fal \tb \}
  \label{eqn:falv-fun-Set} \\
  \tru {\fun \ta \tb}
  &\eqdef \orth{\falv {\fun \ta \tb}}
  = \{ \ea \in \Terms \mid \forall \eb \in \tru{\ta}, \forall \ka \in \fal{\tb}, \mach{\ea}{\cons \eb \ka} \in \Pole \}
  \label{eqn:tru-fun-Set} \\
  \fal {\fun \ta \tb}
  &\eqdef \orth{\tru {\fun \ta \tb}}
  = \biorth{\falv {\fun \ta \tb}}
  \label{eqn:fal-fun-Set}
\end{align}

Truth and falsity witnesses are both defined in terms of a smaller set
$\falv {\fun \ta \tb}$ of \emph{falsity values}, contexts of the
form $\cons \ea \ka$ where $\ea$ is a truth witness for $\ta$, and
$\ka$ a falsity witness for $\tb$.

From a programming point of view, one can think of $\tru \ta$ as a set
of producers of $\ta$, and $\fal \ta$ as a set of consumers of
$\ta$. For example, the falsity values of the arrow type,
$\falv{\fun \ta \tb}$, can be understood as ``to consume a function
$\fun \ta \tb$, one shall produce an $\ta$ and consume a $\tb$''. Logically,
the core argument in a refutation $\fun \ta \tb$ is
a justification of $\ta$ paired with a refutation of $\tb$.
Notice that $\tru{\fun \ta \tb}$ and $\fal {\fun \ta \tb}$ are both
stable by bi-orthogonality, with
$\fal {\fun \ta \tb} = \orth {\tru {\fun \ta \tb}}$ by definition and
$\tru {\fun \ta \tb} \eqset \orth {\fal {\fun \ta \tb}}$ from
$\triorth{S} \eqset \orth{S}$.

\begin{example}
  Additionally, let us consider the following truth values for the
  types of natural numbers and booleans:
\[
  \truv{\tbool} \triangleq \{0, 1\}
\qquad
  \truv{\tnat}  \triangleq \mathbb{N} \: (= \{0, 1, 2, 3, \ldots\}) 
\]

We have:
\vspace{-\baselineskip}

\noindent
\begin{minipage}[t]{.45\textwidth}
\begin{align*}
\ka \in \fal{\tbool}
  &\Leftrightarrow \ka \in \orth {\truv{\tbool}} \\
  &\Leftrightarrow \forall \bva \in \truv{\tbool}, \mach \bva \ka \in \Pole \\
  &\Leftrightarrow \mach 0 \ka \in \Pole \wedge \mach 1 \ka \in \Pole
\end{align*}
\end{minipage}
\begin{minipage}[t]{.45\textwidth}
\begin{align*}
\ea \in \tru{\tnat}
  &\Leftrightarrow \ea \in \biorth {\truv{\tnat}} \\
  &\Leftrightarrow \binder{\forall}{\ka}{(\forall \nva \in \truv{\tnat}, \mach \nva \ka \in \Pole) \to
                                \mach \ea \ka \in \Pole} \\
  &\Leftrightarrow \binder{\forall}{\ka}{(\binder{\forall}{\nva \in \mathbb{N}} {\mach \nva \ka \in \Pole}) \to
                                \mach \ea \ka \in \Pole}
\end{align*}
\end{minipage}

Unfolding definitions,
\begin{align*}
\ea \in \tru{\fun \tnat \tbool}
  &\Leftrightarrow \ea \in \orth {\falv{\fun \tnat \tbool}} \\
  &\Leftrightarrow \binder{\forall}{\sa \in \falv{\fun \tnat \tbool}}{\mach \ea \sa \in \Pole} \\
  &\Leftrightarrow \binder{\forall}{\eb \in \tru{\tnat}}{\binder{\forall}{\ka \in \fal{\tbool}}{\mach \ea {\cons \eb \ka} \in \Pole}}
\end{align*}
we conclude that a term $\ea$ belongs to $\tru{\fun \tnat \tbool}$ if
it goes in the pole whenever we put it against a stack formed of a
term that behaves like a natural number (its argument) and a context
awaiting a Boolean to go into the pole (the continuation to apply to
the output of the function $\ea$).

\end{example}

Realizability interpretations must ensure that provable formulas have
a realizer. This is given by a Fundamental Theorem:
\begin{theorem}[Fundamental Theorem]
  \label{thm:fundamental-thm}
  If $\judgty \, \ea \ta$ then $\transl{\ea} \in \tru \ta$.
\end{theorem}

To prove the fundamental theorem by induction on the typing
derivation, we need to handle open terms (in the abstraction case).
Typing environments
$\ca$ are realized by \emph{substitutions} $\substa$ mapping
term variables to terms as expected:
$
  \substa \in \tru \ca
  \eqdef
  \forall (\eva : \ta) \in \ca,
  \;
  \substa(\eva) \in \tru \ta
$.
The fundamental theorem is then a direct corollary of the more precise
Adequacy Lemma:
\begin{lemma}[Adequacy]
  \label{thm:KAM-adequacy}
  If $\judgty \ca \ea \ta$ then, for any $\substa \in \tru \ca$,
  we have $\transl{\ea}[\substa] \in \tru \ta$.
\end{lemma}
\begin{proof}
  The proof proceeds by induction on the derivation.
  We show the two cases related to the function type.

  \paragraph{Abstraction}
  \begin{mathpar}
    \obinfer
    {\ca, \eva : \ta \der \ea : \tb}
    {\ca \der {\lam \eva \ea} : {\fun \ta \tb}}
  \end{mathpar}
  Our goal is to prove that
  $\obsubst{(\lam \eva \ea)}{\substa} \in \tru {\fun \ta \tb}$.
  Since $\tru {\fun \ta \tb} \eqreas{eqn:tru-fun-Set} \orth{\falv {\fun \ta \tb}}$,
  we have to
  prove that for any context $\cons \eb \ka$ with
  $\eb \in \tru \ta$ and $\ka \in \fal \tb$ we have
  $\mach {(\lam x \ea)[\substa]} {\cons \eb \ka} \in \Pole$. By
  anti-reduction~\eqref{eqn:anti-red} applied to~\eqref{red:lam}, it suffices to show that
  $\mach {\obsubst{\ea}{\substa, \by \eb \eva}} \ka \in \Pole$. The extended
  substitution $\substa, \by \eva \eb$ is in
  $\env {\ca, \eva \of \ta}$, so by induction hypothesis we have that
  $\obsubst \ea {\substa, \by \eb \eva} \in \tru \tb$ -- the body of the function
  behaves well at $\tb$. This suffices to prove
  $\mach {\obsubst \ea {\substa, \by \eb \eva}} \ka \in \Pole$, given that $\ka$
  is in $\fal \tb$, which is precisely the orthogonal of $\tru \tb$.

  \paragraph{Application}
  \begin{mathpar}
    \obinfer
    {\ca \der \ea : {\fun \ta \tb}
     \\
     \ca \der \eb : \ta}
    {\ca \der {\app \ea \eb} : \tb}
  \end{mathpar}
  We want to prove that $(\app \ea \eb)[\substa]$ is in $\tru \tb$,
  that is: given any $\ka \in \fal \tb$, we have
  $\mach {\obsubst{(\app \ea \eb)} \substa} \ka \in \Pole$. By anti-reduction~\eqref{eqn:anti-red} applied to~\eqref{red:app}, it
  suffices to prove
  $\mach {\obsubst \ea \substa} {\cons {\obsubst \eb \substa} \ka} \in \Pole$, which
  we obtain by induction hypothesis: $\obsubst \eb \substa$ is in $\tru \ta$, so
  $\cons {\obsubst \eb \substa} \ka$ is in $\tru \ta \times \fal \tb \eqreas{eqn:falv-fun-Set} \falv {\fun \ta \tb}$ which is included in $\biorth{\falv {\fun \ta \tb}} \eqreas{eqn:fal-fun-Set} \fal {\fun \ta \tb}$ and, since  $\obsubst \ea \substa$ is in
  $\tru {\fun \ta \tb}$, we can conclude that the
  configuration is in the pole.
\end{proof}

We can obtain various meta-theoretic results as consequences of the
fundamental theorem, by varying the choice of the realizability
structure $(\mathbb{T}, \mathbb{E}, \Pole)$. In particular:

\begin{corollary}
  \label{thm:weak-normalization}
  Closed terms in the simply-typed $\lambda$-calculus
  $\judgty{} \ea \ta$ are weakly normalizing.
\end{corollary}

\begin{proof}
  We choose a particular realizability structure by defining $\Terms$
  as the set of closed terms $\ea$ such that $\mach \ea \kstar$ is
  normalizing, $\Coterms$ as the set of closed contexts, and
  $\Pole$ as the set of closed configurations that reach a normal form.
  By definition of $\Terms$, we have that $\kstar$ is in
  $\fal \ta \eqreas{eqn:fal-fun-Set} \orth{\tru \ta}$, for any $\ta$: indeed, for
  any machine term $\ea \in \tru{\ta}$, we have $\tru{\ta} \subset \Terms$
  so, by definition, $\mach \ea \kstar \in \Pole$.

  Let $\ea \in \Terms$ be a closed simply-typed $\lambda$-term and $\der \ea \of \ta$ its typing derivation.
  By the Fundamental Theorem, $\der \ea \of \ta$ implies
  $\transl \ea \in \tru \ta$. Given that $\kstar \in \fal \ta$, we
  have $\mach {\transl \ea} \kstar \in \Pole$, that is,
  $\mach {\transl \ea} \kstar$ reduces to a (closed) normal form
  $\mach \eb \ka$ -- in fact $\eb$ is exactly $\transl \eb$.
  This can only be a normal form if $\eb$
  is a $\lambda$-abstraction and $\ka = \kstar$. We have deduced
  that $\mach {\transl \ea} \kstar \rtos \mach {\transl \eb} \kstar$, so
  we know that $\ea \rtos \eb$ in the $\lambda$-calculus, and $\eb$ is
  a normal form.
\end{proof}

The Fundamental Theorem is also useful to establish various results
based on the shape of normal forms for certain connectives, by varying
the choice of pole and realizability structure. For example, one could
easily prove that closed booleans reduce to $\mathsf{true}$ or
$\mathsf{false}$, or that there is no closed term of the empty
type. See for example \citet*{munch-gdr-normalization}.

In this paper, we prove normalization of the simply-typed
$\lambda$-calculus with sums by showing that the configurations $\ma$
built from well-typed terms are in the pole $\Pole$ of normalizing
configurations. Reading off the computational content from its proof,
we shall discover a normalization function.

\subsection{The \mumutilda{} machine: classical realizability in the positive fragment too}

We mentioned that realizers can be untyped, but they do not
\emph{need} to be untyped. In fact, there is a beautiful way to extend
the type system for the simply-typed $\lambda$-calculus to terms of
the abstract machine, with a judgment for left-introduction rules of the form
$\ob{\ca \mid \ka : \ta \der \tb}$ expressing that the context $\ka$
consumes an input of type $\ta$ to produce an output of type $\tb$.
\begin{mathpar}
\obinfer
{\judgty \ca \eb \ta \\ \ca \mid \ka : \tb \der \tc}
{\ca \mid \cons \eb \ka : \fun \ta \tb \der \tc}

\obinfer
{ }
{\ca \mid \kstar : \ta \der \ta}
\end{mathpar}
This extension appeals to logicians: erasing terms, we recognize the
left-introduction rules of sequent calculus
\begin{mathpar}
\obinfer
{\ca \der \ta \\ \ca , \tb \der \tc}
{\ca , \fun \ta \tb \der \tc}

\obinfer
{ }
{\ca, \ta \der \ta}
\end{mathpar}
%

Unfortunately, there is no direct way to extend the Krivine abstract
machine with sums that would preserve this
correspondence: reduction can get stuck
unless we add commuting conversions.\footnote{Which is why classical
realizability is usually developed in System F, where datatypes are
emulated by Church encoding.}
For our realizability program to bear on the entirety of our source
language, we are thus led to search for a suitable calculus of
realizers first and, next, show how our source language translates to these.

\begin{figure}

\begin{mathpar}
  \begin{array}{c}
    \ma \in \Machines \gramdef \mach \ea \ka \qquad \text{machine} \medskip\\
    \begin{array}{@{}l@{\qquad}l}
      \begin{array}{l@{~}r@{~}l@{\quad}l}
        \ea,\eb \in \Terms
        & \gramdef & & \text{terms} \\
        & \mid & \eva & \text{variable} \\
        & \mid & \bmu{(\cons{x}{\kva})} \ma & \text{abstraction} \\
        & \mid & \inj i \ea & \text{sum} \\
        & \mid & \bmu{\kva} \ma & \text{thunk} \\
        & \mid & \ldots & 
      \end{array}
      &
      \begin{array}{l@{~}r@{~}l@{\quad}l}
        \ka,\kb \in \Coterms
        & \gramdef & & \text{co-terms} \\
        & \mid & \kva & \text{co-variable} \\
        & \mid & \cons \ea \ka & \text{application} \\
        & \mid & \bmutsumfam{\eva}{\ma} & \text{sum elimination} \\
        & \mid & \bmut{\eva}{\ma} & \text{force} \\
        & \mid & \ldots & 
      \end{array}
    \end{array}
  \end{array}
  \end{mathpar}

  \begin{mathpar}
  \infer
  {\mapidx 1 \rto \mapidx 2}
  {\ma[\mapidx 1] \rto \ma[\mapidx 2]}
  \end{mathpar}

\noindent
\begin{minipage}{.4\linewidth}
  \begin{align}
    \mach {\bmu \kva \ma} \ka & \rto \obsubst{\ma}{\by \ka \kva}
        \label{red:mu} \\
    \mach \ea {\bmut \eva \ma} & \rto \obsubst{\ma}{\by \ea \eva}
        \label{red:mut}
  \end{align}
\end{minipage}
\hfill
\begin{minipage}{.55\linewidth}
\begin{align}
    \mach {\bmu {(\cons \eva \kva)} \ma} {\cons \ea \ka}
    & \rto
    \obsubst{\ma}{\by \ea \eva, \by \ka \kva}
        \label{red:mu-cons} \\
    \mach {\inj i \ea} {\bmutsumfam \eva \ma}
    & \rto
    \obsubst{\maidx i}{\by \ea {\evaidx i}}
        \label{red:mut-sig}
  \end{align}
\end{minipage}
%
\caption{An untyped \mumutilda{} abstract machine}
\label{fig:mumutilda}
\end{figure}

Fortunately, the work of~\citet*{curien-herbelin} shows that the bias
toward negativity can be avoided thanks to an abstract machine called
$\mumutilda$. We define its syntax and dynamic semantics
(\figurename~\ref{fig:mumutilda}).\footnote{We use a slight
  improvement of the original syntax, where $\lambda$ is not
  a primitive anymore, due to~\citet*{munch:polarised-lambda-15}.}
Note that we are giving an untyped and unpolarized presentation of
$\mumutilda$. It is known that this presentation is non-confluent but,
as we explained, the language of realizers can have very wild computational
behaviors without issue.
People working with $\mumutilda$ directly, instead of as a language
of realizers, restrict the language to regain confluence by
typing (enforcing normalization), evaluation
strategies~\citep*{downen:tuto-mumutilde} or
polarization~\citep*{munch-phd}.

In the $\mumutilda$-calculus, some terms (or contexts) can capture the
context (or term) set against them. The term $\bmu \kva \ma$ binds its
context to the name $\kva$ and reduces to the machine configuration
$\ma$ (which may contain $\kva$). For example, in the Krivine machine, application $\app \ea \eb$
expects to be put against a context $\ka$, and this reduces to
$\mach \ea {\cons \eb \ka}$. In $\mumutilda$, application
$\app \ea \eb$ is not a primitive term-former but it can be encoded,
following the above reduction principle, as $\bmu \kva {\mach \ea
{\cons \eb \kva}}$. The symmetric construction exists for contexts
$\bmut \eva \ma$ that capture the term set against them.
Contexts are thus more powerful in the \mumutilda{}-machine than they
are in the Krivine machine; they are on an equal footing with
terms. We call them \emph{co-terms}, to express the idea
that they can drive reduction too.

Finally, each connective is defined by a constructor and
a destructor. The constructors for sums are the term
$\inj 1 \ea$ and $\inj 2 \ea$, as in the $\lambda$-calculus. The
destructor $\bmutsumfam \eva \ma$ matches on the sum
constructor, and reduces to one configuration or another depending on
its value. For functions, the constructor is the context former
$\cons{\ea}{\ka}$ -- negative types are defined by their
\emph{observations}, so the constructor builds a co-term. The
destructor, $\bmu {(\cons \eva \kva)} \ma$, matches on the
application context and reduces to a configuration. In particular,
$\lam \eva \ea$ is not a primitive, it can be encoded using the
destructor: $\bmu{(\cons \eva \kva)}{\mach \ea \kva}$.
\figurename~\ref{fig:mumutilda-cbn-compilation} gives a translation scheme
from plain $\lambda$-calculus terms to $\mumutilda$. This is not the
only possible choice -- as we see in Section~\ref{sec:nbr-all},
varying the choice of translation gives different evaluation
strategies.
\begin{figure}
  \begin{mathpar}
  \begin{array}{lll}
    \transl {\lam \eva \ea}
    & \eqdef
    & \bmu {(\cons \eva \kva)} {\mach {\transl \ea} \kva}
    \\
    \transl {\app \ea \eb}
    & \eqdef
    & \bmu \kva {\mach {\transl \ea} {\cons {\transl \eb} \kva}}
    \\
    \transl {\inj i \ea}
    & \eqdef
    & \inj i {\transl \ea}
    \\
    \transl {\desumfam \ea \eva \eb}
    & \eqdef
    & \bmu \kva {\mach {\transl \ea} {\vbmutsumlam {\eva_\idx} {\mach {\transl {\ebidx \idx}} \kva}}}
  \end{array}
  \end{mathpar}
  \caption{A compilation scheme from the $\lambda$-calculus to $\mumutilda$ (fragment with functions and sums)}
  \label{fig:mumutilda-cbn-compilation}
\end{figure}

\begin{example}[Commuting conversion]
Consider the term
$$\app {\pdesumlam \evc {\eva_\idx} {\lam \evb {\ea_\idx}}} \eb$$
In
the $\lambda$-calculus, the function application outside is stuck, and
unlocking it would require adding a commuting conversion to push it
under the case-split. In $\mumutilda$ (as in the sequent-calculus),
the translation of this term reduces without a hurdle (redex locations are underlined):
\[
  \begin{array}{@{}r@{\,}c@{\,}l@{}}
    & & \transl{\app {\pdesumlam \evc {\eva_\idx} {\lam \evb {\ea_\idx}}} \eb}
    \\
    & =
    & \bmu \kva {\mach
      {\bmu {\underline{\kvb}} {\mach {\strut \transl \evc}
                        {\bmutsumlam {\eva_\idx}
                           {\mach {\bmu {(\cons \evb \kvc)} {\mach {\transl {\ea_\idx}} \kvc}}  \kvb}}}}
      {\underline{\cons {\transl \eb} \kva}}}
    \\[.5em]
    & \rto^{(\ref{red:mu})}
    & \bmu \kva {\mach {\transl \evc}
      {\bmutsumlam {\eva_\idx}
      {\mach {\bmu {(\underline{\cons \evb \kvc})} {\mach {\transl {\ea_\idx}} \kvc}} {\underline{\cons {\transl \eb} \kva}}}}}
    \\[.5em]
    \text{\scriptsize (2 $\times$)} & \rto^{(\ref{red:mu-cons})}
    & \bmu \kva {\mach {\strut \transl \evc}
      {\bmutsumlam {\eva_\idx}
      {\mach {\transl {\obsubst {\ea_\idx} {\by \eb \evb}}} \kva}}}
    \\
  \end{array}
\]
\end{example}

For the Krivine machine we mention in Section~\ref{sub:rea-negative}
a simulation property: if $\mach t e$ reduces to a normal machine
$\mach u \kstar$, $\lambda$-term $u$ is the normal form of $t$ for the
weak call-by-name strategy. For the $\mumutilda$ calculus, a similar
result holds: if $\mach {\transl \ea} \alpha$ reduces to a normal
machine $\mach \eap \alpha$, with then $\eap$ can be easily translated
back to a $\lambda$-term that is equivalent to $\ea$ and is (at least)
in weak-head-normal form. Proving this, using for example results from
\citet*{munch:polarised-lambda-15}, requires building more technical
knowledge of $\mumutilda$ than what is presented in this article; here
we focus on computing the normal machine from
$\mach {\transl \ea} \alpha$.

\subsection{Interpreting types with abstract machines}
\label{sub:witnesses}
\label{sec:types}

In Section~\ref{sub:rea-negative}, we demonstrated classical
realizability for the simply-typed lambda-calculus with just function
types. To scale the argument to our whole calculus, we define for each
type $\ta$ a truth witness $\tru\ta \subseteq \Terms$ (justifications, or producers)
and a falsity witness $\fal\ta \subseteq \Coterms$ (refutations, or consumers).
They are orthogonal to each other, $\orth{\fal \ta} \eqset \tru \ta \label{eqn:fal-tru}$ and $\orth{\tru \ta} \eqset \fal \ta \label{eqn:tru-fal}$, and defined in terms of a more primitive set of truth or falsity \emph{value witnesses}, $\truv \ta \subseteq \Terms$ or $\falv \ta \subseteq \Coterms$.
More precisely, positive type formers (such as sums and integers) are defined by their truth value witnesses, whereas negative type formers (such as functions and the unit type) are defined by their falsity witnesses.

\begin{multicols}{2}
  \noindent
  \begin{align}
    \falv{\fun \ta \tb}
    & \eqdef \left\{ \cons \eb \ka \mid \eb \in \tru \ta, \ka \in \fal \tb \right\}
      \label{eqn:falv-fun} \\
    \tru{\fun \ta \tb}
    & \eqdef \orth{\falv{\fun \ta \tb}}
      \label{eqn:tru-fun} \\
    \fal{\fun \ta \tb}
    & \eqdef \biorth{\falv{\fun \ta \tb}} = \orth{\tru {\fun \ta \tb}}
      \label{eqn:fal-fun} \\
    \falv{\tunit}
    & \eqdef \varnothing
      \label{eqn:falv-unit} \\
    \tru{\tunit}
    & \eqdef \orth{\falv{\tunit}}
      \label{eqn:tru-unit} \\
    \fal{\tunit}
    & \eqdef \biorth{\falv{\tunit}} = \orth{\tru{\tunit}}
      \label{eqn:fal-unit} \\
    \truv \tna & \eqdef \tru \tna
    \label{eqn:truv-neg}
  \end{align}
  \begin{align}
    \truv{\sum \ta \tb}
    & \eqdef \left\{ \inj i \ea \mid \ea \in \tru{\ta_i} \right\}
      \label{eqn:truv-sum} \\
    \fal{\sum \ta \tb}
    & \eqdef \orth{\truv{\sum \ta \tb}}
      \label{eqn:fal-sum} \\
    \tru{\sum \ta \tb}
    & \eqdef \biorth{\truv{\sum \ta \tb}} = \orth{\fal{\sum \ta \tb}}
      \label{eqn:tru-sum} \\
    \truv{\tnat}
    & \eqdef \mathbb{N}
      \label{eqn:truv-nat} \\
    \fal{\tnat}
    & \eqdef \orth{\truv{\tnat}}
      \label{eqn:fal-nat} \\
    \tru{\tnat}
    & \eqdef \biorth{\truv{\tnat}} = \orth{\fal{\tnat}}
      \label{eqn:tru-nat} \\
    \falv \tpa & \eqdef \fal \tpa
    \label{eqn:falv-pos}
  \end{align}
\end{multicols}

As we remarked before, for any set~$S$ we have
$\biorth{(\orth{S})} \eqset \orth{S}$. As a result, truth and falsity
witnesses, defined as orthogonals or bi-orthogonals of value
witnesses, are stable by bi-orthogonality.

\XGS{This extension is necessary to type-check the implementation
  function-application cases (simplified and dependent) as we have
  written them.} The two definitions $\eqref{eqn:truv-neg}$ and
$\eqref{eqn:falv-pos}$ extend the notions of (truth and falsity) value
witnesses to both positive and negative types. They give the property
that we have $\tru \ta = \orth {\falv \ta}$ for any $\ta$, not just
negative types, and conversely $\fal \ta = \orth {\truv \ta}$. It will
be convenient when implementing value witnesses of these types as programs.

\section[Normalization by realizability for the simply-typed lambda-calculus]%
        {Normalization by realizability for the simply-typed $\lambda$-calculus}
\label{sec:nbr-variant1}
\label{subsec:realization-program}

Having set up the realizability framework in the previous Section, we
now cast the corresponding adequacy lemma as a dependently-typed
program (Section~\ref{sec:adequacy}). Its implementation will be the
subject of Section~\ref{subsec:dependent}.
Before getting into deep(endently-typed!) water, we simplify our model
(Section~\ref{subsec:simplification}) so as to focus solely on the
computational content of realizability, at the expense of its logical
content. In Section~\ref{sec:simply-typed-adequacy}, we implement a
realizability-based proof of weak normalization of the simply-typed
$\lambda$-calculus with arrows and sums in this simplified
setting.

The inductive definitions of syntactic terms, types and type
derivations that we gave in \figurename~\ref{fig:lambda} should be
understood as being inductive definitions in ambient type theory in
which we express the adequacy lemma: this program manipulates
syntactic representations of terms $\ea$, types $\ta$, and well-formed
derivations (dependently) typed by a judgment $\judgty \ca \ea \ta$.
To make this clear, we stick to the color conventions introduced in
Section~\ref{sec:colors}, which take all their meaning here as we
combine programming constructs in the object and the meta-langagues:
for example, we distinguish the $\lambda$-abstraction of the object
language $\lam \eva \ea$ (a piece of data representing
a simply-typed abstraction) and of the meta-language
$\semlam \semeva \semea$ (a program).

\subsection{Adequacy}
\label{sec:adequacy}

To distill the computational content of realizability, we implement
the adequacy lemma \emph{as a program} in a rather standard type
theory. Martin-Löf type theory with one predicative universe (denoted $\Type$) suffices for our
development. The universe is required to define the type of truth and
falsity witnesses by recursion on the syntax of object-language types.
In our Coq development, we also exploit the sort $\Prop$ as a means to
filter out computationally-irrelevant terms from the extracted
proof/program but we do not exploit its impredicativity.

Truth witnesses (and, respectively, falsity witnesses and the
pole) contain \emph{closed} programs, which have no free variables. To
interpret an open term $\judgty \ca \ea \ta$, one should first be passed
a substitution $\substa$ that maps the free variables of $\ea$ to closed
terms. Such a substitution is \emph{compatible with the typing environment
$\ca$} iff for each mapping $\eva \mathrel{\ob{\of}} \ta$ in $\ca$, the substitution
maps $\eva$ to a truth witness for $\ta$: $\substa(\eva) \in \tru{\ta}$. We use
$\env{\ca}$ to denote the set of substitutions satisfying this property.

In this setting, the adequacy lemma amounts to a function of the
following type:
\begin{mathpar}
  \rea : \binder{\forall}{\impl{\ca}\,\ea\,\impl{\ta}\,\impl{\substa}}{
    \semfun{\ \impl{\judgty \ca \ea \ta}}{
      \semfun{\ \substa \in \env{\ca}\ }{
        \ \obsubst{\transl{\ea}}{\substa} \in \tru{\ta}}}}
\end{mathpar}
As such, the result of this program is a proof of a set-membership
$\transl{\ea} \in \tru{\ta}$, whose computational status is unclear (and
is not primitively supported in most type theories). Our core idea is
to redefine these types in a proof-relevant way: we treat the
predicate $\_ \in \tru{\_} : \Terms \to \type \to \Type$ as a type of witnesses where the
inhabitants of $\transl{\ea} \in \tru{\ta}$ are the different syntactic
justifications that the \mumutilda-term $\transl{\ea}$ indeed
realizes $\ta$.
We respect a specific naming convention for arguments whose
dependent type is an inhabitation property: an hypothesis of type
$\ea \in \tru{\ta}$ is named $\semea$, $\semva$ for the
type $\va \in \truv{\ta}$ of \emph{\textbf{v}alues}, $\semka$
for the type $\ka \in \fal{\ta}$ of
co-terms, $\semsa$ for the type
$\sa \in \falv{\ta}$ of \emph{linear co-terms}\footnote{Linear
co-terms are to co-terms what values are to terms. They correspond
to ``evaluation'' contexts that force the opposite term to be evaluated
exactly once, justifying the adjective ``linear''.
\XGS{Attention, les co-termes $\pi$ ne vérifient pas
$\forall t, \mach t \pi \not\rto$, prendre par exemple $t \eqdef \bmu \kva {m'}$.}
The notation
$\pi$ comes from the similarity between the French pronunciations of
$\pi$ and stack: $\pi$ is pronounced [pi] whereas
stack, \textit{\textbf{pi}le} in French, is pronounced [pil].}, and
finally $\semsubsta$ for the type $\substa \in \env{\ca}$.
Names like $\eva$, $\ea$ and $\eb$, $\ka$, $\ta$ and $\tb$, $\ca$
are used to name respectively syntactic term variables, terms, co-terms, types, and
contexts.

The adequacy result is parametric on the pole, but we analyze its computational behavior
with a particular definition of the pole in mind, or rather of its proof-relevant membership
predicate
$\_ \in \Pole : \Machines \to \Type$ defined as
\begin{equation}\label{eqn:pole-Set}
 \ma \in \Pole
   \quad \eqdef \quad
   \{ \ob{\ma_n} \in \Values \mid \ma \rto \maidx 1 \rto \dots \rto \ma_n  \}
\end{equation}
where $\Values$ is the set of \emph{normal configurations}, \ie{},
configurations that cannot reduce further but are not stuck on an
error: they are of the form $\mach \eva \ka$ or $\mach \ea \kva$
but not, for example, $\mach {\inj i \ea} {\cons \ea \ka}$.
With this definition, an adequacy lemma proving $\ma \in \Pole$ for
a well-typed $\ma$ is exactly a normalization program -- the question
is \emph{how} it computes.
In this setting, the orthogonality predicates $\_ \in \orth{\tru{\_}} : \Coterms \to \type \to \Type$
and $\_ \in \orth{\fal{\_}} : \Terms \to \type \to \Type$ admit a straight-forward proof-relevant
definition as dependent function types:
\begin{align}
  \ka \in \orth{\tru{\ta}} &\eqdef
  \binder{\forall}{\{\ea : \Terms\}}%
    {\semfun{\ea \in \tru{\ta}}{\mach \ea \ka \in \Pole}}
    \label{eqn:orthT-SetD}
\\
  \ea \in \orth{\fal{\ta}} &\eqdef
  \binder{\forall}{\{\ka : \Coterms\}}%
    {\semfun{\ka \in \fal{\ta}}{\mach \ea \ka \in \Pole}}
    \label{eqn:orthF-SetD}
\end{align}

\paragraph{Notation} Following common usage, we write
$\binder{\forall}{\impl{\semeva : \ta}}{\tb}$ to declare an implicit
argument of type $\ta$, meaning that we omit those arguments when
writing calls to function of such type, as they can be non-ambiguously
deduced from the other arguments or the ambient environment. We write
$\semlam{\impl{a}}{t}$ if we want to explicitly bind an implicit
argument. On paper, this remains an informal notation meant to reduce the
syntactic overhead of dependent types: we shall not concern ourselves with the constraints of typability here, our Coq development having the final
say.

\subsection{Simplification}
\label{subsec:simplification}

In this section, we introduce a simpler notion of pole that is not
as orthodox, but results in simpler programs.
This pedagogical detour allows us to focus exclusively on the
computational content of the adequacy lemma. We leave it out to
Section~\ref{subsec:dependent} to present a correct-by-construction
adequacy lemma, where we are then bound to grapple with the
computational \emph{and} logical content in one fell swoop.
Here, we simply define
$\ma \in \Pole \eqdef \Values$: a witness that $\ma$ is well-behaved
is not a reduction sequence to a normal configuration $\man$, but only that
configuration $\man$. This is a weaker type, as it does not guarantee that the
value we get in return of the adequacy program is indeed obtained from
$\ma$ (it may be any other normal configuration). In
Section~\ref{subsec:dependent}, the more informative definition of
$\Pole$ is reinstated, and we show how the simple programs
we are going to write can be enriched to keep track of this reduction
sequence -- incidentally demonstrating that they were indeed returning
the correct value.

A pleasant side-effect of this simplification is that the set
membership types are not dependent anymore: the definition of $\ma \in
\Pole$ does not depend on $\ma$ ; definitions of $\va \in \truv{\ta}$, $\ea \in \tru{\ta}$, $\sa \in \falv{\tb}$, $\ka
\in \fal{\tb}$, and $\substa \in \env{\ca}$ do not mention $\va$, $\ea$, $\sa$, $\ka$, nor $\substa$ ; and orthogonality is defined with non-dependent types.

To make that simplification explicit, we rename those types
$\juPole$, $\jutruv{\ta}$, $\jutru{\ta}$, $\jufalv{\tb}$, $\jufal{\tb}$ and
$\juenv{\ca}$: they are \textbf{j}ustifications of membership
of some configuration (the type system does not track which one), term, co-term or context to the respective set.

Recall from Section~\ref{sub:witnesses} the definitions of value witnesses for the type formers:
\[
\begin{array}{r@{\:}l@{\quad}l}
  \falv {\fun \ta \tb} &\eqreas{eqn:falv-fun}
  \left\{ \cons \eb \ka \mid \eb \in \tru \ta, \ka \in \fal \tb \right\} \\
  \falv {\tunit} &\eqreas{eqn:falv-unit} \varnothing
\end{array}
\begin{array}{r@{\:}l@{\quad}l}
  \truv {\sum \ta \tb} &\eqreas{eqn:truv-sum}
  \left\{ \inj{i}{\ea} \mid \ea \in \tru{\taidx i} \right\}
  & (\ob{i} \in \{\ob{1}, \ob{2}\}) \\
  \truv {\tnat} &\eqreas{eqn:truv-nat} \mathbb{N} \\
\end{array}
\]
Filtering out their logical content, the Cartesian product inside $\falv{\fun \ta
  \tb}$ simplifies to a type of pairs (of the meta-language) whereas the union inside $\truv{\sum \ta \tb}$ becomes a sum type (of the meta-language):

\begin{minipage}{0.5\linewidth}
  \begin{align}
  \jufalv{\fun \ta \tb}
  &\eqdef
  \semprod {\jutru{\ta}} {\jufal{\tb}}
  \tag{\ref*{eqn:falv-fun}'} \label{eqn:falv-fun-J} \\
    \jufalv \tunit & \eqdef \varnothing
    \tag{\ref*{eqn:falv-unit}'}\label{eqn:falv-unit-J}
  \end{align}
\end{minipage}
\hfill
\begin{minipage}{0.45\linewidth}
  \begin{align}
  \jutruv{\sumfam \ta}
  &\eqdef
  \semsum {\jutru{\taidx 1}} {\jutru{\taidx 2}}
  \tag{\ref*{eqn:truv-sum}'}\label{eqn:truv-sum-J} \\
    \jutruv \tnat & \eqdef \mathbb{N}
    \tag{\ref*{eqn:truv-nat}'}\label{eqn:truv-nat-J}
  \end{align}
\end{minipage}
\smallskip

Similarly, the definition of orthogonality becomes:

\begin{minipage}{0.45\linewidth}
  \begin{align}
    \juPole & \eqdef \Values
      \tag{\ref*{eqn:pole-Set}'}\label{eqn:pole-J}
  \end{align}
\end{minipage}
\hfill
\begin{minipage}{0.45\linewidth}
  \begin{align}
    \juorth{T} & \eqdef \semfun T \juPole
     \tag{\ref*{eqn:orthT-SetD}'}\label{eqn:orth-J}
  \end{align}
\end{minipage}

\medskip \noindent
from which we derive the interpretation of types, through the usual
closure construction

\noindent
\begin{minipage}{.45\linewidth}
\begin{align}
  \jutru{\fun \ta \tb} &\eqdef \juorth{\jufalv{\fun \ta \tb}}
      \tag{\ref*{eqn:tru-fun}'}\label{eqn:tru-fun-J} \\
  \jufal{\fun \ta \tb} &\eqdef \jubiorth{\jufalv{\fun \ta \tb}}
      \tag{\ref*{eqn:fal-fun}'}\label{eqn:fal-fun-J} \\
  \jutru{\tunit} &\eqdef \juorth{\jufalv{\tunit}}
      \tag{\ref*{eqn:tru-unit}'}\label{eqn:tru-unit-J} \\
  \jufal{\tunit} &\eqdef \jubiorth{\jufalv{\tunit}}
      \tag{\ref*{eqn:fal-unit}'}\label{eqn:fal-unit-J} \\
  \jutruv \tna & \eqdef \jutru \tna
      \tag{\ref*{eqn:truv-neg}'}\label{eqn:truv-neg-J}
\end{align}
\end{minipage}
\hfill
\begin{minipage}{.5\linewidth}
\begin{align}
  \jufal{\sum \ta \tb}
      &\eqdef \juorth{\jutruv{\sum \ta \tb}}
      \tag{\ref*{eqn:fal-sum}'}\label{eqn:fal-sum-J} \\
  \jutru{\sum \ta \tb} &\eqdef \jubiorth{\jutruv{\sum \ta \tb}}
      \tag{\ref*{eqn:tru-sum}'}\label{eqn:tru-sum-J} \\
  \jufal{\tnat}
      &\eqdef \juorth{\jutruv{\tnat}}
      \tag{\ref*{eqn:fal-nat}'}\label{eqn:fal-nat-J} \\
  \jutru{\tnat} &\eqdef \jubiorth{\jutruv{\tnat}}
      \tag{\ref*{eqn:tru-nat}'}\label{eqn:tru-nat-J} \\
  \jufalv \tpa & \eqdef \jufal \tpa
      \tag{\ref*{eqn:falv-pos}'}\label{eqn:falv-pos-J}
\end{align}
\end{minipage}
\smallskip

Note that the types $\jutruv \ta$, $\jufalv \ta$
are \emph{not} to be understood as inductive types indexed by an
object-language type $\ta$, for they would contain fatal recursive
occurrences in negative positions:
\begin{mathpar}
\begin{array}{r@{\ =\ }l}
  \jufalv{\ob{\fun{(\fun \ta \tb)} \tc}}
  & \semprod{\jutru{\fun \ta \tb}}{\jufal{\tc}} \\
  & \semprod{\juorth{\jufalv{\fun \ta \tb}}}{\jufal{\tc}} \\
  & \semprod{(\semfun{\jufalv{\fun \ta \tb}}{\juPole})}{\jufal{\tc}} \\
\end{array}
\end{mathpar}
Instead, one should interpret $\jutruv{\ta}$ and $\jufalv{\ta}$ as
mutually recursive type-returning \emph{functions} defined by
structural induction over (the syntactic structure of) their input $\ta$.

\subsection{A simply-typed realizability program}
\label{sec:simply-typed-adequacy}

With this in place, we can write our simplified adequacy lemma, in
the non-dependent version:
\begin{mathpar}
  \rea : \binder{\forall}{\impl{\ca}\,\ea\,\impl{\ta}\,\impl{\substa}}{
      \ \impl{\judgty \ca \ea \ta}\ \to \juenv{\ca} \to \jutru{\ta}}
\end{mathpar}

We present the code case by case, and discuss why each is well-typed. On paper, we focus our attention to the type system presented in \figurename~\ref{fig:lambda}, ignoring integers and the unit.
To help the reader type-check the code, we write
$M^{\tru{\ob S}}$ (resp., $M^{\fal{\ob S}}$) if the expression $M$ has type
$\jutru{\ob S}$ (resp., $\jufal{\ob S}$).
For example, $\semea^{\tru{\ta}}$ can be interpreted as $\semea :
\jutru{\ta}$.

\paragraph{Variable}
\begin{mathpar}
  \begin{array}{rlllll}
    \rea & \eva
      & \impl{\raisebox{-.1em}{$\ob{\infer{(\eva \of \ta) \in \ca}{\judgty \ca \eva \ta}}$}}
      & \semsubsta^{\env{\ob{\ca, \eva : \ta}}}
      & \eqdef
      & \semsubsta(\eva)^{\tru{\ta}}
      \\
    \end{array}
\end{mathpar}
By hypothesis, the binding $\ob{x \of A}$ is in the context $\ob{\ca, \eva \of A}$, and we
have $\semsubsta : \juenv{\ob{\ca, \eva \of A}}$, so in particular
$\semsubsta(\eva) : \eva \in \tru{\ta}$ as expected.

\paragraph{Abstraction}
\label{sec:abstraction-simply-typed}
\begin{mathpar}
  \begin{array}{rlllll}
    \rea & (\ob{\lam{x}{t}})
      & \impl{\ob{
        \infer{\judgty{\ca, \eva \of \ta} \ea \tb}
              {\judgty \ca {\lam \eva \ea} {\fun \ta \tb}}
        }}
      & \semsubsta^{\env{\ca}}
      & \eqdef
      & \semlam{(\semeb^{\tru{\ta}}, \semka^{\fal{\tb}})}{
          \cut{\tb}
              {\rea~\ea~\semsubsta[\eva \mapsto \semeb]^{\env{\ob{\ca, \eva : \ta}}}}
              {\semka}
        }
      \\
    \end{array}
\end{mathpar}

We have structural information on the return type:
\[
  \jutru{\ob{\fun \ta \tb}}
  \eqreas{eqn:tru-fun-J}
  \juorth{\jufalv{\ob{\fun \ta \tb}}}
  \eqreas{eqn:orth-J}
  \semfun{\jufalv{\ob{\fun \ta \tb}}}{\juPole}
  \eqreas{eqn:falv-fun-J}
  \semfun{\jutru{\ta} \times \jufal{\tb}}{\juPole}
\]
It is thus natural to start with a $\lambda$-abstraction over
$\jutru{\ta} \times \jufal{\tb}$
, matching a pair of an
$\semeb : \jutru{\ta}$ and an $\semka : \jufal{\tb}$ to
return a $\juPole$.

The recursive call on $\ob{t : B}$ gives a witness of type
$\jutru{\tb}$; we combine it with a witness of type
$\jufal{\tb}$ to give a $\juPole$ by using an auxiliary
cut function which gives control to the active part of the term: either $\semea$ if $\ta$ is a positive type ($\tnat$, or a sum), or $\semka$ if $\ta$ is a negative type ($\tunit$, or a function).
In the simply-typed setting, this cut function boils down to
\begin{equation} 
  \label{eqn:cut-J}
  \begin{array}{lcll}
    \cut{\ta}{\_}{\_}
    & : &
    \jutru{\ta} \to
    \jufal{\ta} \to
    \juPole
    \\
    \cut{\tna}{\semea}{\semka}
    & \eqdef
    & \semapp{\semka}{\semea}
    & (\tna \in \{ \tunit, \fun \ta \tb \})
    \\
    \cut{\tpa}{\semea}{\semka}
    & \eqdef
    & \semapp{\semea}{\semka}
    & (\tpa \in \{ \tnat, \sum \ta \tb \})
    \\
  \end{array}
\end{equation}

\paragraph{Injections}

\begin{mathpar}
  \begin{array}{rlllllr}
    \rea & (\ob{\inj{i}{\ea}})
      & \impl{\ob{\infer
                         {\judgty \ca \ea {\taidx i}}
                         {\judgty \ca {\inj{i}{\ea}} {\sumfam{A}}}}}
      & \semsubsta
      & \eqdef
      & \biorth{(\inj{i}{(\rea~\ea~\semsubsta)})}
      & (i \in \{1, 2\})
      \\
  \end{array}
\end{mathpar}

The injection constructor applied to the recursive call has type
$\semsum{\jutru{\taidx 1}}{\jutru{\taidx 2}}
    \eqreas{eqn:truv-sum-J}
\jutruv{\ob{\sumfam{A}}}$.
This is not the expected type of
the $\rea$ function, which is
$\jutru{\ob{\sumfam{A}}}
    \eqreas{eqn:tru-sum-J}
  \jubiorth{\jutruv{\ob{\sumfam{A}}}}$.
The solution is to appeal to the bi-orthogonal inclusion $S \subseteq \biorth{S}$, which boils down to the following (pair of) simply-typed program, which turns value witnesses into witnesses, akin to a CPS translation:
\begin{equation}
  \label{eqn:biorth-J}
  \begin{array}{lcl}
  \val{\_}
  &:& \jufalv{\tna} \to \jufal{\tna} \\
  \val{(\semsa^{\falv{\tna}})}
  &\eqdef& \semlam{\semea^{\tru{\tna}}}{\semapp{\semea}{\semsa}} \\
  && (\tna \in \{\tunit, \fun \ta \tb\})
  \end{array}
\qquad
  \begin{array}{lcl}
  \val{\_}
  &:& \jutruv{\tpa} \to \jutru{\tpa} \\
  \val{(\semva^{\truv{\tpa}})}
  &\eqdef& \semlam{\semka^{\fal{\tpa}}}{\semapp{\semka}{\semva}} \\
 &&(\tpa \in \{ \tnat, \sum \ta \tb \}) \\
 \end{array}
\end{equation}

\paragraph{Application}
\begin{mathpar}
  \begin{array}{rlllll}
    \rea & (\ob{\app{\ea}{\eb}})
      & \impl{\ob{
        \infer{\judgty{\ca}{\ea}{\fun \ta \tb} \\ \judgty \ca  \eb \ta}
              {\judgty \ca {\app \ea \eb} \tb}
        }}
      & \semsubsta
      & \eqdef
      & \semlam{\semsa^{\falv{\tb}}}{
          \rea~\ea~\semsubsta
              ~\sempair{\rea~\eb~\semsubsta}{\val{\semsa}}
        } \\
  \end{array}
\end{mathpar}

Unlike the previous cases, we know nothing about the structure of the
expected return type $\jutru{\tb}$, so it seems unclear at
first how to proceed. In such situations, the trick is to use the
fact that
$\tru{\tb} = \orth{\falv{\tb}}$:
the expected type is therefore
$\juorth{\falv{\tb}}
    \eqreas{eqn:orth-J}
  \jufalv{\tb} \to \juPole$.

The function parameter $\semsa$ at type $\jufalv{\tb}$ is
injected into $\jufal{\tb}$ by the $\val{\_}$ function
defined in~\eqref{eqn:biorth-J}, and then paired with a recursive call on $\eb$ at
type $\jutru{\ta}$ to build
a
$\semprod{\jutru{\ta}}
         {\jufal{\tb}}
    \eqreas{eqn:falv-fun-J}
  \jufalv{\ob{\fun \ta \tb}}$.

This co-value justification can then be directly applied to the
recursive call on $\ea$, of type
$\jutru{\ob{\fun \ta \tb}}
    \eqreas{eqn:tru-fun-J}
  \juorth{\jufalv{\ob{\fun \ta \tb}}}
    \eqreas{eqn:orth-J}
  \jufalv{\ob{\fun \ta \tb}} \to \juPole
    \eqreas{eqn:falv-fun-J}
  \jutru{\ta} \times \jufal{\tb} \to \juPole$.

\paragraph{Case analysis}

\begin{mathpar}
    \rea ~ {\pdesumfam \ea \eva \eb} ~
      \impl{\obinfer{\judgty \ca \ea {\sumfam \ta}\\
                     \judgty {\ca, \evaidx i \of \taidx i} {\ebidx i} \tc}
                    {\judgty \ca {\desumfam \ea \eva \eb} \tc}}
      ~ \semsubsta
      ~ \eqdef

        \semlam{\semsa^{\falv{\tc}}}
            {\cut{\sumfam \ta}
                 {\rea~\ea~\semsubsta}
                 {\semlam{\semva^{\truv{\sumfam{\ta}}}}{
                     \semdesumlam{\semva}{{\semeaidx \idx}^{\tru{\taidx \idx}}}{
                       \rea~\ebidx \idx
                           ~\semsubsta[\evaidx \idx \mapsto \semeaidx \idx]
                           ~\semsa
                     }}}}
\end{mathpar}

Here we use one last trick: it is not in general possible to turn a
witness in $\tru{\ta}$ into a value witness in $\truv{\ta}$ (respectively, a
witness in $\fal{\tb}$ into a value witness in $\falv{\tb}$), but it is when the
return type of the whole expression is $\juPole$: we can
cut our $\semea : \jutru{\ta}$ with an abstraction
$\semlam{\semva^{\truv{\ta}}}{\dots}$ (thus providing a name $\semva$ at type
$\jutruv{\ta}$ for the rest of the expression) returning a
$\juPole$, using the cut function:
$
\cut{\ta}{\semea}{\semlam{\semva}{M^{\Pole}}} : \juPole
$. (If you know about monads, this is the \emph{bind} operation.)
%






\begin{figure}

\begin{mathpar}
  \begin{array}{rllllr}
    \rea & \ob{x^A}
      & \semsubsta
      & \eqdef
      & \semsubsta(\eva)
      \\
    \rea & \ob{(\lam{x^A}{t^B})}
      & \semsubsta
      & \eqdef
      & \semlam{(\semeb^{\tru{\ta}}, \semka^{\fal{\tb}})}{
          \cut{\tb}
              {\rea~\ea~\semsubsta[\eva \mapsto \semeb]}
              {\semka}
        }
      \\
    \rea & \ob{(\app{t^{\fun \ta \tb}}{u^{A}})}
      & \semsubsta
      & \eqdef
      & \semlam{\semsa^{\falv{\tb}}}{
          \rea~\ea~\semsubsta
              ~\sempair{\rea~\eb~\semsubsta}{\val{\semsa}}
        } \\
    \rea & \ob{(\inj{i}{t^{A_i}})}
      & \semsubsta
      & \eqdef
      & \biorth{(\seminj{i}{(\rea~\ea~\semsubsta)})} & (i \in \{1, 2\})
      \\
    \rea & \ob{\pdesumlam{\ea^{\sumfam{A}}}{\eva_\idx}{{\eb_\idx}^{C}}}
         & \semsubsta
         & \eqdef
         &
\end{array}

           \semlam{\semsa^{\falv{\tc}}}
            {\cut{\ob{\sumfam{\ta}}}
                 {\rea~\ea~\semsubsta}
                 {\semlam{\semva^{\truv{\sumfam{\ta}}}}{\semdesumlam{\semva}
                          {{\semeaidx \idx}^{\tru{\taidx \idx}}}
                          {\rea~\ob{\ebidx \idx}
                            ~\semsubsta[\evaidx \idx \mapsto \semeaidx \idx]
                          ~\semsa
                        }}}}
\end{mathpar}

\caption{Summary of the simply-typed setting}
\label{fig:variant1}
\end{figure}

\smallskip
This concludes our implementation of the adequacy lemma. For the sake
of completeness, its complete definition is summarized
in \figurename~\ref{fig:variant1}.  It should be clear at this point
that the computational behavior of this proof is, as expected, a
normalization function. The details of how it works, however, are
rather unclear to the non-specialist, in part due to the relative
complexity of the types involved: the call $\rea~\ea~\semsubsta$
returns a function type (of type $\semfun{\jufalv{\ta}}{\juPole}$), so it may not evaluate its argument
immediately. We further dwell on this question in
Section~\ref{sec:nbr-all}, in which we present a systematic study of
the various design choices available in the realizability proof.
As for now, we move from the simplified, non-dependent types to the
more informative dependent types. This calls for defining the
compilation from $\lambda$-terms to \mumutilda-machines, which has
been left unspecified so far.

\section{A dependently-typed realizability program}
\label{subsec:dependent}
\label{subsec:forma}

In the following, we undo the simplification presented in
Section~\ref{subsec:simplification}, by moving back to dependent
types. The best way to make sure that our program is type-correct is
to run a type-checker on it. To this end, this section has been
developed under the scrutiny of the Coq proof assistant. Working with
Coq offers two advantages. First, it has allowed us to interactively
explore the design space in a type-driven manner, using type-level
computation to assist this exploration. Second and unlike the
set-theoretic presentation of Section~\ref{sub:witnesses}, we can
precisely delineate the computational and propositional content of the
proof.


The interpretation of types rests, as usual, on a suitable definition
of truth and falsity value witnesses. We translate the set-theoretic
definitions~\eqref{eqn:falv-fun}, \eqref{eqn:falv-unit}, \eqref{eqn:truv-sum}
and \eqref{eqn:truv-nat} with the dependently-typed functions
$\_ \in \falv{\_} : \Coterms \to \type \to \Type$ and
$\_ \in \truv{\_} : \Terms \to \type \to \Type$ specified by
\begin{minipage}{0.45\linewidth}
\begin{align*}
\sa \in \falv{\fun \ta \tb}
  &\eqdef
  \filter{\sa}
         {\cons \eb \ka}
         {\eb \in \tru{\ta} \times \ka \in \fal{\tb}}
\tag{\ref*{eqn:falv-fun}''}\label{eqn:falv-fun-D}
\\
\sa \in \falv{\tunit}
  &\eqdef \bot
\tag{\ref*{eqn:falv-unit}''}\label{eqn:falv-unit-D}
\end{align*}
\end{minipage}
\hfill
\begin{minipage}{0.5\linewidth}
\begin{align*}
\va \in \truv{\ob{\sumfam \ta}}
  &\eqdef
    \filtertwo{\va}
              {\inj{1}{\eb}}
              {\eb \in \tru{A_1}}
              {\inj{2}{\eb}}
              {\eb \in \tru{A_2}}
\tag{\ref*{eqn:truv-sum}''}\label{eqn:truv-sum-D}
\\
\va \in \truv{\tnat}
  &\eqdef
    \filter{\va}
           {\enat n \in \mathbb{N}}
           {\top}
\tag{\ref*{eqn:truv-nat}''}\label{eqn:truv-nat-D}
\end{align*}
\end{minipage}
\XGS{Note that in the $\truv \tnat$ case here we assume that
  $\enat n \in \mathbb{N}$ belongs to $\Terms$, that is to
  $\mumutilda$ terms, while we ommitted the constants from the
  grammar. For now we can live with this small discrepancy.}

Applying the usual closure construction with the pole~\eqref{eqn:pole-Set} and orthogonality predicates \eqref{eqn:orthT-SetD} and \eqref{eqn:orthF-SetD}, we tie the knot by recursion over types and obtain a mutually
recursive pair of functions from syntactic types of the object
language to predicates in our meta-language: a truth witness
interpretation $\_ \in \tru{\_} : \Terms \to \type \to \Type$ building predicates over $\Terms$ and a falsity
witness interpretation $\_ \in \fal{\_} : \Coterms \to \type \to \Type$ building predicates over $\Coterms$. This
naturally extends to contexts.


At this stage, the adequacy lemma has grown into a dependently-typed program of signature
$$
  \rea : \binder{\forall}{
    \impl{\ca}\, \ea\,\impl{\ta}\,\impl{\substa}
  }{
    \ \impl{\ob{\Gamma \der \ty t  A}}\ \to
      \substa \in \env{\ca} \to
      \transl{\ea}[\substa] \in \tru{\ta}
  }
$$
where $\transl{\ea}$ compiles a term of the
object language into a \mumutilda-term in the meta-language that
witnesses the reduction. 
Rather than define this function upfront, we are going to
reverse-engineer it from the proof of the adequacy lemma. Indeed, the
typing constraints of the dependent version force us to exhibit a
reduction sequence, that is, a suitable compilation function.

The dependent version of the adequacy lemma has the same structure as the non-dependent one. Apart from enriching types, the main difference is that we justify the existence of the desired reduction sequence through computationally transparent annotations tracking the use of reduction rules.

\paragraph{Abstraction}
Consider the $\ob{\lambda}$-abstraction case of Section~\ref{sec:nbr-variant1} (p.~\pageref{sec:abstraction-simply-typed}):
\begin{mathpar}
  \begin{array}{rlllll}
    \rea & \ob{(\lam{x^A}{t^B})^{\fun \ta \tb}}
      & \semsubsta
      & \eqdef
      & \semlam{(\semeb, \semka)^{\falv{\ob{\fun \ta \tb}}}}{
          \cut{\tb}{\rea~\ea~\semsubsta[\eva \mapsto \semeb]}{\semka}
        }
      \\
    \end{array}
\end{mathpar}
We transform this code to a  dependently-typed case, by adding
annotations, leaving the code structure otherwise unchanged:
\begin{mathpar}
    \rea
    ~ \impl{\ca}
    ~ \ob{(\lam{x^A}{t^B})^{\fun \ta \tb}}
    ~ \impl{\substa}
    ~ (\semsubsta : \substa \in \env{\ca})
    ~ \eqdef

    \semlam{\impl{\ob{\cons{u}{e}} : \Coterms}}{
       \semlam{(\sempair{\semeb : \eb \in \tru{\ta}}
                  {\semka : \ob{e} \in \fal{\tb}}
             : \ob{\cons{u}{e}} \in \falv{\ob{\fun \ta \tb}})}{
       \ %
       \cut{\tb}{\rea~\ea~\semsubsta[\eva \mapsto \semeb]}{\semka}^\mathtt{lam}
     }}
\end{mathpar}

Note that the implicit abstraction is on inputs of the form
$\cons \eb \ka$ rather than a general input $\kaidx 0 : \Coterms$.
Indeed, it is followed by an abstraction over an element of type
$\kaidx 0 \in \falv {\fun \ta \tb}$, which by definition
\eqref{eqn:falv-fun-D} would be empty unless $\kaidx 0$ is of the form
$\cons \eb \ka$.

The cut function~(\ref{eqn:cut-J}) can be extended to a dependent type without any change in its implementation:
\begin{mathpar}
\cut{\ta}{\_}{\_} : \forall\,\impl{\ea\,\ob{e}},\,
   \ea \in \tru{\ta} \to
   \ob{e} \in \fal{\ta} \to
   \ob{\mach{t}{e}} \in \Pole
\end{mathpar}
but this is not enough to make the whole term type-check. Indeed, this
gives to the expression
$\cut{\tb}{\rea~\ea~\semsubsta[\eva \mapsto \semeb]}{\semka}$
the type $\mach{\obsubst{\transl{\ea}}{\substa, \eva \mapsto \eb}}{\ka} \in \Pole$.
But we just abstracted on $\cons \eb \ka \in \falv{\fun \ta \tb}$, in
order to form a complete term of type
$\ob{\obsubst{\transl{\lam{x}{t}}}{\substa}} \in \orth{\falv{\ob{\fun \ta \tb}}}$
so the expected type is in fact $\ob{\mach{\obsubst{\transl{\lam{x}{t}}}{\substa}}{\cons{u}{e}}} \in
\Pole$. To address this issue, we define an auxiliary function  $\_^\mathtt{lam}$  of type
\begin{mathpar}
  \_^\mathtt{lam} : \binder{\forall}{\impl{\ea\,\substa\,\eva\,\eb\,\ka}}{
    \mach {\obsubst {\transl{\ea}} {\substa, \eva \mapsto \eb}} \ka \in \Pole
    \to
    \mach {\obsubst {\transl{\lam{\eva}{\ea}}} \substa} {\cons{\eb}{\ka}} \in \Pole
  }
\end{mathpar}

This auxiliary function exactly corresponds to the fact that
the pole is closed under anti-reduction of machine configurations
(Figure~\ref{fig:mumutilda}). Indeed, the required reduction
\[
  \mach {\transl {\lam \eva \ea}} {\cons \eb \ka}
  \rto
  \mach {\obsubst \ea {\eva \mapsto \eb}} \ka
\]
is a consequence of the general rule for functions
\[
  \mach {\bmu {(\cons \eva \kva)} {\mach \ea \kva}} {\cons \eb \ka}
  \rto^{(\ref{red:mu-cons})}
  \mach {\obsubst \ea {\eva \mapsto \eb, \kva \mapsto \ka}} \ka
\]
plus the fact that translations of $\lambda$-terms $\transl \ea$
have no free co-term variable $\kva$
-- as can be proved by direct induction.
Note that if we had \emph{not} given you the definition of $\transl {\lam \eva \ea}$,
you would be \emph{asked} to pick $\bmu {(\cons \eva \kva)} {\mach \ea \kva}$
from the type-checking constraints that arise in the proof.

To define this function, recall our definition of the pole:
\begin{mathpar}
\ma \in \Pole
   \quad \eqdef \quad
   \{ \ma_n \in \Values \mid \ma \rto \maidx 1 \rto \dots \rto \ma_n \}
\end{mathpar}
The function $\_^\mathtt{lam}$ takes a normal form for the reduced
machine, and has to return a normal form for the not-yet-reduced
machine; this is the exact same normal form, with an extra reduction step:
\begin{mathpar}
  \begin{array}{lll}
    \_^\mathtt{lam} ~ {\impl{\ea\,\substa\,\eva\,\eb\,\ka}}
    & \eqdef &
      \lam {\bo{\{ \ob{\ma_n} \mid
      \mach {\obsubst{\transl{\ea}}{\eva \mapsto \eb}} \ka
      \rto \dots \rto \ob{\ma_n} \}}} {}
    \\
    & & \ \ %
    \bo{\{ \ob{\ma_n} \mid
      \mach {\transl{\lam{\eva}{\ea}}} {\cons{\eb}{\ka}}
      \rto \mach {\obsubst{\transl{\ea}}{\eva \mapsto \eb}} \ka
      \rto \dots \rto \ob{\ma_n} \}}
  \end{array}
\end{mathpar}

\paragraph{Application}
Similarly, the application case can be dependently typed as
\begin{mathpar}
    \rea
      ~ \ob{(\app{\ea^{\fun \ta \tb}}{\eb^{\ta}})}
      ~ \semsubsta
      ~ \eqdef
      ~
    \semlam{\impl{\sa : \Coterms}}{
      \semlam{\semsa : \sa \in \falv{\tb}}{
        {\left(\semapp{\rea~\ea~\semsubsta}
                      {{\sempair{\rea~\eb~\semsubsta}
                                {\val{\semsa}}}}\right)}^\mathtt{app}
    }}
\end{mathpar}

The function $\val{\_}$, which we had defined at the type
$\jufalv{\ta} \to \jufal{\ta}$, can be given -- without
changing its implementation -- the more precise type
$\binder{\forall}{\impl{\sa}}{\sa \in \falv{\ta} \to \sa \in \fal{\ta}}$
. This means that
$\semapp{\rea~\ea~\semsubsta}
        {\sempair{\rea~\eb~\semsubsta}{\val{\semsa}}}
$
has type
$
\mach
   {\obsubst{\transl{\ea}}{\substa}}
   {\cons{\obsubst{\transl{\eb}}{\substa}}{\sa}}
 \in \Pole$.
For
the whole term to be well-typed, the auxiliary function $\_^\mathtt{app}$
needs to have the type
\[
\binder{\forall}{\impl{\ea\,\eb\,\substa\,\sa}}{
  \mach{\obsubst {\transl \ea} \substa}
       {\cons{\obsubst {\transl \eb} \substa}{\sa}}} \in \Pole
  \to
  \ob{\mach{\obsubst {\transl{\app{\ea}{\eb}}} \substa}{\sa}} \in \Pole
\]

Again, this is exactly an anti-reduction rule; not as a coincidence, but
as a direct result of the translation of application in $\mumutilda$:
$\transl {\app \ea \eb}
 \eqdef
 \bmu \kva {\mach {\transl \ea} {\cons {\transl \eb} \kva}}$.
The definition is similar to $\_^\mathtt{lam}$.

\paragraph{Case analysis}
The dependently-typed treatment of case analysis involves three decorations,
${\_}^{\mathtt{inj}_1}$, ${\_}^{\mathtt{inj}_2}$
 and ${\_}^{\mathtt{case}}$:
\[
\begin{array}{l}
    \rea ~ \pddesumlam{\ea^{\sumfam{A}}}{\eva_\idx}{{\eb_\idx}^{C}}
         ~ \semsubsta
         ~ \eqdef
  \\ \quad
    \semlam{\impl{\sa : \Coterms}}{
      \semlam{\semsa : \sa \in \falv{\tc}}{
  }}
  \\ \qquad
  \newcommand{\arm}[1]{%
    {(\rea~\ebidx #1
          ~\semsubsta[\evaidx #1 \mapsto \semeaidx #1]
          ~\semsa)}^{\mathtt{inj}_{#1}}}
   \cut{\sumfam{\ta}}
       {\rea~\ea~\semsubsta}
       {\semlam{\impl{\va}}{\semlam{\semva : \va \in \truv{\sumfam \ta}}{
        \sembrawcase{\va, \semva}
          {\inj 1 {\eaidx 1}, \seminj 1 {\semeaidx 1}}
          {\arm 1}
          {\inj 2 {\eaidx 2}, \seminj 2 {\semeaidx 2}}
          {\arm 2}}}}^{\mathtt{case}}
\end{array}
\]

We can illustrate how the typing constraints in this term
lets us deduce/recover exactly the translation we gave
in \figurename~\ref{fig:mumutilda-cbn-compilation}
of sum elimination into $\mumutilda$.

For readability, we use here the more compact family
notation $\pfamdesumfam \ea \iot i \eva \eb$ for
$\pddesumfam \ea \eva \eb$.

From the outer typing constraint, we know that the result type of $\_^\mathtt{case}$
must be of type
$\mach {\obsubst {\transl {\famdesumfam \ea \iot i \eva \eb}} \substa} \sa
 \in \Pole$.
Its argument is of the form $\cut {\sumfam \ta} {\rea~\ea~\substa} \dots$,
so $\_^\mathtt{case}$ must perform an anti-reduction argument on a reduction step
of the form
\[
  \mach {\transl {\famdesumfam \ea \iot i \eva \eb}} \sa
  \rto
  \mach {\transl \ea} \dots
\]

Let us write $(\transl {\famdesumfam \square \iot i \eva \eb}; \sa)$
for the co-term in the right-hand-side of the reduction rule above.
From the desired reduction step above, we know that we need the translation
\[
  \transl {\famdesumfam \ea \iot i \eva \eb}
  \eqdef
  \bmu \kva {\mach \ea {\bo{(}\transl {\famdesumfam \square \iot i \eva \eb}; \kva\bo{)}}}
\]
where $(\transl {\famdesumfam \square \iot i \eva \eb}; \kva)$ remains to be determined.

As a second step, we look at the typing constraints
for the inner auxiliary functions
$\_^{\mathtt{inj}_i}$, in the terms
$
{(\rea~\ebidx i
  ~\semsubsta[\evaidx i \mapsto \semeaidx i]
~\semsa)}^{\mathtt{inj}_i}
$.
We must prove that the co-term
\(\transl {\famdesumfam \square \iot i \eva \eb}; \sa\)
is in $\fal{\sumfam \ta}$, that is, that putting it against any \emph{constructor}
$\inj i {\eaidx i}$ goes in the pole. The result type of $\_^{\mathtt{inj}_i}$
is thus of the form
\(
  \mach {\inj i {\eaidx i}} {\transl {\famdesumfam \square \iot i \eva \eb}; \sa}
  \in \Pole
\)

Finally, we know the input type of $\_^{\mathtt{inj}_i}$, which is of the form
$
\mach {\obsubst {\ebidx i} {\by {\evaidx i} {\eaidx i}}} \sa \in \Pole
$
so we know that this function must correspond to an anti-reduction argument
for the reduction
\[
\mach {\inj i {\eaidx i}} {\transl {\famdesumfam \square \iot i \eva \eb}; \sa}
\rto
\mach {\obsubst {\ebidx i} {\by {\evaidx i} {\eaidx i}}} \sa
\]
which forces us to pose:\footnote{We took the idea of seeing
  introduction of a new destructor as the
  resolutions of some equations from \citet*{munch-phd}.}
\[
  (\transl {\famdesum \square \iot i \eva \eb}; \sa)
  \eqdef
  \bmutsumlam {\evaidx \idx} {\mach {\ebidx \idx} \sa}
\]
which gives, indeed, the same definition as
\figurename~\ref{fig:mumutilda-cbn-compilation}:
\begin{align*}
  \transl {\famdesumfam \ea \iot i \eva \eb}
  & =
  \mach \ea {\transl {\famdesumfam \square \iot i \eva \eb}; \sa} \\
  & =
  \mach \ea {\bmutsumlam {\evaidx \idx} {\mach {\ebidx \idx} \sa}}
\end{align*}

\paragraph{Sum injections, integers and unit type}
Injections, being computationally inert, are straightforward to
handle: an injection in the object language compiles into
a \mumutilda{} injection and the corresponding realizability program
does not require decoration since it does not need to track
a reduction of the configuration.
Integers and the unit type are also computationally inert and even simpler since there is no subterm to normalize: the compilation is simply the identity and the correctness proof is the diagonal.

\paragraph{Wrapping up} We would like to point out that all of the compilation
rules from the $\lambda$-calculus to our abstract machines can be \emph{recovered}
by studying the typing constraints in our proof.

It means that this compilation scheme is, in some sense, solely
determined by the type of the adequacy lemma. In particular, we
recovered the fact that the application transition needs to be
specified only for linear co-terms $\sa$ -- which, in our
setting, correspond to those co-terms that belong not only to
$\fal{\fun \ta \tb}$, but also to $\falv{\fun \ta \tb}$.

\paragraph{Mechanized formalization in Coq}
Our formalization covers more type constructors than we had space to
present here, handling both (positive) natural numbers and positive
products\footnote{In intuitionistic logic, there is little difference
between the positive and the negative product, so a product is always
suspect of being negative type in hiding.  The elimination form of
positive products matches strictly on both arguments whereas negative
products are eliminated by two projections. The difference is glaring
in the presence of side-effects.}. While positive products are handled
along the same line as sums, interpreting the elimination form for
natural numbers lead us naturally to introduce a standalone lemma that
extracts (somewhat to our disbelief) to an iterator over integers.
This phenomena by which we coincidentally extract to the right
program is noteworthy. In fact, there is little to no choice in
writing the proof, once we get familiarized with the mechanics of our
dependent types.

Our implementation takes advantage of distinctive form of the pole: it
consists in a computationally relevant object (the normal form
$\ob{\ma_n} \in \Values$) together with a proposition stating that the
normal form is indeed related to the original program. By threading
this invariant through our model construction by means of subset types, the dependently-typed
adequacy lemma -- proved using tactics and concluded with the
vernacular \texttt{Defined.} -- admits a straightforward extraction:
the second component of the subset types is erased, leaving only the
normal forms. The resulting OCaml program is almost literally the
program written in Section~\ref{sec:nbr-variant1}, provided that we
ignore the remains of dependent types that extraction fails to
eliminate, such as terms, stacks, and substitutions over them.  In
particular, we uncovered during this process that the reduction
sequence is justified by appeal to computationally irrelevant
annotations translating the anti-reduction closure of the pole, a
reassuring and satisfying result.  An informal dead-code elimination
argument suggests that we could safely remove the remains of dependent
types but convincing Coq that it is safe to do so remains
future work.

To illustrate the fact that the computational content of the adequacy
lemma does not depend on a particular choice of the pole, the proof
itself is parametrized by a Coq module specifying the computational
and logical content of a pole.  We provide two
instantiations of pole to support extraction: full machine configurations and natural
numbers (suited for configurations that normalize to integers). In this
way, we are able to run several examples of derivations in the
simply-typed $\lambda$-calculus through the adequacy lemma and extract
the corresponding normalized configuration or integer. The
supplementary material provides information on how to run
with these proofs/programs.

\section{Alternative realizations}
\label{sec:nbr-all}

We made two arbitrary choices when we decided to define
$\falv {\fun \ta \tb} \eqdef
    \{ \cons \eb \ka \mid \eb \in \tru \ta, \ka \in \fal \tb \}$
. There
are in fact four possibilities with the same structure:
\begin{mathpar}
  \variant{1}\quad \{ \cons \eb \ka \mid \eb \in \tru \ta, \ka \in \fal \tb \}

  \variant{2}\quad \{ \cons \eb \ka \mid \eb \in \tru \ta, \ka \in \falv \tb \}

  \variant{3}\quad \{ \cons \eb \ka \mid \eb \in \truv \ta, \ka \in \fal \tb \}

  \variant{4}\quad \{ \cons \eb \ka \mid \eb \in \truv \ta, \ka \in \falv \tb \}
\end{mathpar}

In our Coq formalization, we study the four possibilities and see that
they all work but each of them gives a slightly different realization
program. For conciseness, we only treat Variant~\variant{3} in this
section and omit sum types, integers and the unit type;
they can be handled separately (and sums would also give rise to four different
choices), as we did in our formalization.

We do not repeat the reverse engineering process of the previous
section, but directly jump to the conclusion of what translation to
(arrow-related) co-terms they suggest, and which reduction strategy
they must follow. We observe that the first two variants correspond to
a call-by-name evaluation strategy, while the two latter correspond to
call-by-value; forcing us, in particular, to use the $\bmut\eva \ma$
construction~\citep*{curien-herbelin} in the compilation to
$\mumutilda$-terms.

In order to distinguish the compilation functions of the various variants, we mark them with their variant number as exponent: $\transl[i]{\ea}$ for variant~\variant{i}.

\paragraph{Variant~\variant{1} $\{ \cons \eb \ka \mid \eb \in \tru \ta, \ka \in \fal \tb \}$}
This interpretation corresponds to the realization program presented
in Section~\ref{subsec:realization-program}, recalled here for ease of
comparison:
\begin{mathpar}
  \begin{array}{rllll}
    \rea & \eva^{\ta}
      & \semsubsta
      & \eqdef
      & \semsubsta(\eva)
      \\
    \rea & {(\lam {\eva^{\ta}} {\ea^{\tb}})}^{\fun \ta \tb}
      & \semsubsta
      & \eqdef
      & \semlam{(\semeb^{\tru{\ta}}, \semka^{\fal{\tb}})}{
          \cut{\tb}{\rea~\ea~\semsubsta[\eva \mapsto \semeb]}{\semka}
        }
      \\
    \rea & \ob{(\app {\ea^{\fun \ta \tb}} {\eb^{\ta}})}
      & \semsubsta
      & \eqdef
      & \semlam{\semsa^{\falv{\tb}}}{
          \rea~\ea~\semsubsta
              ~\sempair{\rea~\eb~\semsubsta}{\val{\semsa}}
        } \\
    \end{array}

\begin{array}{rl}
\transl[1]{\eva}
 &\eqdef \eva \\
\transl[1]{\lam \eva \ea}
 &\eqdef \bmu{(\cons \eva \kva)}{\mach{\transl[1]{\ea}}{\kva}} \\
\transl[1]{\app \ea \eb}
 &\eqdef \bmu{\kva}{\mach{\transl[1]{\ea}}{\cons{\transl[1]{\eb}}{\kva}}} \\
\end{array}
\end{mathpar}

\paragraph{Variant~\variant{3} $\{ \cons \eb \ka \mid \eb \in \truv \ta, \ka \in \fal \tb \}$}

A notable difference in this variant -- and Variant~\variant{4} -- is
that we can restrict the typing of our environment witnesses to store
value witnesses: rather than $\semsubsta : \substa \in \env{\ca}$, we
have $\semsubsta : \substa \in \envv{\ca}$, that is, for each binding
$\ob{\ea \of \ta}$ in $\ca$, we assume $\semsubsta(\eva) : \substa(\eva) \in \truv{\ta}$.
The function would be typable with a weaker argument
$\semsubsta : \substa \in \env{\ca}$ (and would have an equivalent
behavior when starting from the empty environment), but that would
make the type less precise about the dynamics of evaluation.

\abovedisplayskip=1em
\belowdisplayskip=1em

\begin{mathpar}
  \begin{array}{rllll}
    \rea & \eva^{\ta}
      & \semsubsta^{\envv{\ca}}
      & \eqdef
      & \val{(\semsubsta(\eva))}
      \\
    \rea & (\lam {\eva^{\ta}} {\ea^{\tb}})^{\fun \ta \tb}
      & \semsubsta
      & \eqdef
      & \semlam{(\semva^{\truv{\ta}}, \semka^{\fal{\tb}})}{
          \cut{\tb}{\rea~\ea~\semsubsta[\eva \mapsto \semva]}{\semka}
        }
      \\
    \rea & \ob{(\app{\ea^{\fun \ta \tb}}{\eb^{\ta}})}
      & \semsubsta
      & \eqdef
      & \semlam{\semsa^{\falv{\tb}}}{
          \cut{\ta}{\rea~\eb~\semsubsta}{\semlam{\semva_u^{\truv{\ta}}}{
              \cut{\fun \ta \tb}{\rea~\ea~\semsubsta}
                  {\sempair{\semva_u}{\val{\semsa}}}
          }}
        } \\
    \end{array}
\end{mathpar}

Had we kept $\semsubsta : \substa \in \env{\ca}$, we would
have only $\semsubsta(\eva)$ in the variable case, but
$\semsubsta[\eva \mapsto \val{\semva}]$ in the abstraction case,
which is equivalent if we do not consider other connectives pushing in
the environment.

It is interesting to compare the application case with the one of
Variant~\variant{1}:
\begin{mathpar}
\semlam{\semsa^{\falv{\tb}}}{
          \rea~\ea~\semsubsta
              ~\sempair{\rea~\eb~\semsubsta}{\val{\semsa}}
        }
\end{mathpar}
The relation between Variant~\variant{1} and Variant~\variant{3} seems
to be an inlining of $\rea~\eb~\semsubsta$, yet the shapes of the
terms hint at a distinction between call-by-name and call-by-value
evaluation. These two versions are equivalent whenever $\ta$ is
negative (a function), as in that case the cut
$\cut{\ta}{\rea~\eb~\semsubsta}{\semlam{\semva_u}{\dots}}$
applies its left-hand side as an argument to its right-hand side,
giving exactly the definition of Variant~\variant{1} after
$\beta$-reduction. However, once again, considering the dependent
version forces us to clarify the evaluation dynamics:
\[
\begin{array}{@{}l}
    \rea
      ~ \impl{\ca}
      ~ \ob{(\app{t^{\fun \ta \tb}}{u^{A}})}
      ~ \impl{\substa}
      ~ (\semsubsta : \substa \in \envv{\ca})
      ~ \eqdef
      \semlam{\impl{\ob{\pi} : \Coterms}}
     {\semlam{\semsa : \ob{\pi} \in \falv{\tb}}{}}
\\
      \qquad
      \cut{\ta}{\rea~\eb~\semsubsta}{
        \semlam{\impl{\ob{v_u} : \Terms}}{\semlam{(\semva_u : \ob{v_u} \in {\truv{\ta}})}{
          \;
          \cut{\fun \ta \tb}{\rea~\ea~\semsubsta}
              {\sempair{\semva_u}{\val{\semsa}}}^\mathtt{app-body}
        }}
      }^\mathtt{app-arg}
\end{array}
\]

The annotation $\_^\mathtt{app-body}$ corresponds to a reduction to
the configuration $\ob{\mach{\obsubst{\transl[3]{\ea}}{\substa}}{\cons{v_u}{\pi}}}$. Its return type --
which determines the configuration which should reduce to this one -- is
constrained not by the expected return type of the $\rea$ function as
in previous examples, but by the typing of the outer cut function
which expects a $\ob{\obsubst{\transl[3]{\eb}}{\substa}} \in \tru{\ta}$ on the left,
and thus on the right an $\ob{e} \in \fal{\ta}$ for some co-term $\ob{e}$.
By definition of $\ob{e} \in \fal{\ta}$, the return type of
$\_^{\mathtt{app-body}}$ should thus be
of the form $\ob{\mach{v_u}{e}} \in \Pole$.

What could be the result of $\transl[3]{\eb}$ that would respect the reduction
equation
$\ob{\mach{v_u}{\transl[3]{\eb}}} \rto \ob{\mach{\obsubst{\transl[3]{\eb}}{\substa}}{\cons {v_u} \pi}}$
so that $\_^{\mathtt{app-body}}$ amounts to stability
under anti-reduction?
Using the $\ob{\mutilda}$
binder \citep*{curien-herbelin} from $\mumutilda$, we define
$\transl[3]{\eb}
    \eqdef
        \ob{\binder{\mutilda}{v_x}{\mach{\obsubst{t}{\substa}}{\cons {v_x} \pi}}}$
which is subject to the desired reduction:
$\ob{\mach{v}{\binder{\mutilda}{x}{\ma}}} \rto^{(\ref{red:mut})} \ob{\obsubst{\ma}{\eva \mapsto \va}}$.

From here, the input and output type of $\_^{\mathtt{app-arg}}$ are
fully determined, allowing to state a final equation:
$\ob{\mach{\transl[3]{\app \ea \eb}}{\pi}}
\rto
\ob{\mach{\transl[3]{\eb}}{\bmut{v_u}{\mach{\transl[3]{\ea}}{\cons{v_u}{\pi}}}}}$
that is solved by taking
$\transl[3]{\app \ea \eb}
    \eqdef
        \bmu{\kva}
            {\mach{\transl[3]{\eb}}
                  {\bmut{v_u}{\mach{\transl[3]{\ea}}
                                              {\cons{v_u}{\kva}}}}}$
. Summing up:
\begin{mathpar}
\begin{array}{rl}
\transl[3]{\eva}
 &\eqdef \eva \\
\transl[3]{\lam \eva \ea}
 &\eqdef \bmu{(\cons \eva \kva)}{\mach{\transl[3]{\ea}}{\kva}} \\
\transl[3]{\app \ea \eb} &\eqdef
        \bmu{\kva}
            {\mach{\transl[3]{\eb}}
                  {\bmut{v_u}{\mach{\transl[3]{\ea}}
                                              {\cons{v_u}{\kva}}}}}
\end{array}
\end{mathpar}
which corresponds to the compilation scheme from the call-by-value $\lambda$-calculus
into \mumutilda.

\paragraph{A closer look at reductions}

Looking at these two variants, the code itself may not be clear enough
to infer their evaluation order. However, their compilation
to \mumutilda{} machines, that were imposed on us by the dependent
typing requirement, are deafeningly explicit. When interpreting the
function type $\falv{\fun \ta \tb}$, requiring truth \emph{value}
witnesses for $\ta$ gives us call-by-value transitions
(Variant~\variant{3}), while just requiring arbitrary witnesses gives us
call-by-name transitions (Variant~\variant{1}).


Finally, an interesting point to note is that we have explored the
design space of the semantics of the arrow connective independently
from other aspects of our language, notably its sum type. Any of these
choices for $\falv{\fun \ta \tb}$ may be combined with any choice
for $\truv{\sum \ta \tb}$ (basically, lazy or strict sums; the
Coq formalization has a variant with lazy sums) to form a complete
proof.\footnote{With the exception that it is only possible to
  strengthen the type of environments from $\env{\ca}$ to
  $\envv{\ca}$ if all computation rules performing
  substitution only substitute values; that is a global change.} We
see this as yet another manifestation of the claimed modularity of
realizability and logical-relation approaches, which allows studying
connectives independently from each other -- once a notion of
computation powerful enough to support them all has been fixed.

\section{Related work}

\paragraph{Classical realizability}

Realizability has been used by logicians to enrich a logic with new
axioms and constructively justify such extensions. To do so, one
extends the language of machines, terms and co-terms with new
constructs and reduction rules, in order to give a computational
meaning to interesting logical properties which had no witnesses in
the base language. For example, extending the base calculus with a
suitably-defined $\mathtt{callcc}$ construct provides a truth witness
for Pierce's law (and thus, provides a computational justification to
the addition of the principle of excluded middle in the source
logic). When possible, this gives us a deeper (and often surprising)
understanding of logical axioms. This paper takes the opposite
approach, starting from a logic (embodied by the simply-typed
$\lambda$-calculus) and an interpretation of its typed values to
re-discover the computational content of its realizers.

\paragraph{Intuitionistic realizability}

\citet*{oliva-streicher} factor the use of classical realizability for
a negative fragment of second-order logic as the composition of
intuitionistic realizability after a CPS translation. We know that
translating $\lambda$-terms into abstract machines then back
corresponds to a CPS transform. In this work, the authors remark that
classical realizability corresponds to translating the abstract
machines back into (CPS-forms of) $\lambda$-terms, then doing intuitionistic
realizability on the resulting terms. The bi-orthogonal closure arises
from the intuitionistic realizability interpretation of the
double-negation translation of intuitionistic formulas (the types of
terms in the original machine).

From the point of view of a user wishing to understand how
realizability techniques prove normalization of a given well-typed
term, there are thus two choices: either see this term as part of an
abstract machine language, and look at the classical realizability
proof on top of it, or look at the CPS translation of this term and
look at the intuitionistic realizability proof on top of it. Oliva and
Streicher show that those two interpretations agree; the question is
then to know which gives the better explanation. In our opinion,
the abstract machine approach, by giving a direct syntax for
continuations instead of a functional encoding, make it easier to
follow what is going on. In particular, the justification of the
difference in treatment between negative and positive types seems
particularly clean in such a symmetric setting; it is also easy to see
the relation between changes to the configuration reductions and the
reduction strategy of the proof. This is precisely the term-level
counterpart of the idea that sequent calculus is more convenient than
natural deduction to reason about cut-elimination in presence of
positives. That said, because our meta-language is also a typed
$\lambda$-calculus, some aspects of the CPS-producing transformation from
machines to $\lambda$-terms are also present in our system.

\citet*{ilik} proposes a modified version of intuitionistic
realizability to work with intuitionistic sums, by embedding the CPS
translation inside the definition of realizability. The resulting
notion of realizability is in fact quite close to our classical
realizability presentation, with a ``strong forcing'' relation
(corresponding to our set of \emph{value} witnesses), and
a ``forcing'' relation defined by double-negation ($\sim$
bi-orthogonal) of the strong forcing. In particular, Danko Ilik
remarks that by varying the interplay of these two notions (notably,
requiring a strong forcing hypothesis in the semantics of
function types), one can move from a call-by-name interpretation
to a call-by-value setting, which seems to correspond to our
findings.

\paragraph{Normalization by Evaluation}

There are evidently strong links with normalization by evaluation
(NbE). The previously cited work of \citet*{ilik} constructs, as is
now a folklore method, NbE as the composition of a soundness proof
(embedding of the source calculus into a mathematical statement
parametrized on concrete models) followed by a strong completeness
proof (instantiation of the mathematical statement into the syntactic
model to obtain a normal form). In personal communication, Hugo
Herbelin presented orthogonality and NbE through Kripke models as two
facets of the same construction. Orthogonality focuses on the
(co)-terms realizing the mathematical statement, exhibiting an untyped
reduction sequence. NbE focuses on the typing judgments; its statement
guarantees that the resulting normal form is at the expected type, but
does not exhibit any relation between the input and output term.


\paragraph{System L}

It is not so surprising that the evaluation function exhibited in
Section~\ref{subsec:realization-program} embeds a CPS translation,
given our observation that the definition of truth and falsity value witnesses
determines the evaluation order. Indeed, the latter fact implies that
the evaluation order should remain independent from the evaluation
order of the meta-language (the dependent $\lambda$-calculus used to
implement adequacy), and this property is usually obtained by CPS or
monadic translations.

System L is a direct-style calculus that also has this property of
forcing us to be explicit about the evaluation order -- while being
nicer to work with than results of CPS-encoding. In particular,
\citet*{munch-csl} shows that it is a good calculus to study classical
realizability. Our different variants of truth/falsity witnesses
correspond to targeting different subsets of System L, which also
determine the evaluation order. The principle of giving control of
evaluation order to the term or the co-term according to polarity is
also found in Munch-Maccagnoni's PhD thesis~\citep*{munch-phd}.

The computational content of adequacy for System L has been studied,
on the small $\mu\mutilda$ fragment, by Hugo Herbelin in his
habilitation thesis~\citep*{herbelin-hdr}. The reduction strategy is
fixed to be call-by-name. We note the elegant regularity of the
adequacy statements for classical logic, each of the three versions
(configurations, terms and co-terms) taking an environment of truth
witnesses (for term variables) and an environment of falsity witnesses
(for co-term variables).

Another, more compositional way to understand the results presented in
this paper
consists in performing a \emph{typed} translation from the
$\lambda$-calculus to typed System L, where we have a typing judgment
for terms of the form $\judgty \ca \ea \ta$, a typing judgment for
co-terms of the form $\ob{\ca \mid \ka : \ta \der \kva : \ta}$ and well-typed configurations
$\ob{\ma : (\ca \der \kva : \ta)}$ are the pair of a well-typed term, and a well-typed
co-term, interacting along the same type $\ta$. This translation amounts
to a CPS transform, System L being essentially a direct syntax for writing
programs in CPS.

This type system is also backed by an adequacy lemma formed of three
mutually recursive results of the following form:
\begin{itemize}
\item for any $\ea$ and $\substa$ such as $\judgty \ca \ea \ta$
      and $\substa \in \env{\ca}$,
      we have $\obsubst \ea \substa \in \tru{\ta}$
\item for any $\ka$ and $\substa$ such as $\ob{\ca \mid \ka : \ta \der \kva : \ta}$,
  $\substa \in \env{\ca}$ and $\kb \in \fal{\ta}$, we have $\obsubst \ka {\substa, \by \kb \kva} \in \fal{\ta}$
\item for any $\ma$ and $\substa$ such as $\ob{\ma : (\ca \der \kva : \ta)}$,
   $\substa \in \env{\ca}$ and $\kb \in \fal{\ta}$, we have $\obsubst \ma {\substa, \by \kb \kva} \in \Pole$
\end{itemize}

A careful translation from the $\lambda$-calculus to System L would have
allowed us to piggy-back on the adequacy lemma of System L to obtain the
desired adequacy lemma on $\lambda$-terms. In fact, we saw in
Section~\ref{sec:nbr-all} that our formalization is currently
organized around a choice of translation from the $\lambda$-terms to $\mumutilda$:
the missing piece consists in going from $\mumutilda$ to System L, which
amounts to specifying the polarity of our translation -- or, put
otherwise, deciding on the strictness/laziness of the constructs of
the language. Such a translation goes beyond the intent of the present
paper, which is meant as providing a slow-paced transition from approaches based on
$\lambda$-calculus to ones based on abstract machines in meta-theoretical works.

\paragraph{Realizability in PTS}
\citet*{bernardy-lasson} have some of the most intriguing work on
realizability and parametricity we know of. In many ways they go above
and beyond this work: they capture realizability predicates not as an
ad-hoc membership, but as a rich type in a pure type system (PTS) that
is built on top of the source language of realizers -- itself
a PTS. They also establish a deep connection between realizability and
parametricity -- we hope that parametricity would be amenable to
a treatment resembling ours, but that is purely future work.

One thing we wanted to focus on was the \emph{computational content}
of realizability techniques, and this is not described in their work;
adequacy is seen as a mapping from a well-typed term to a realizer
inhabiting the realizability predicate. But it is a meta-level
operation (not described as a \emph{program}) described solely as an
annotation process. We tried to see more in it -- though of course, it
is a composability property of denotational model constructions that
they respect the input's structure and can be presented as
trivial mappings (e.g.,
$\sembrackets{\ob{\app{t}{u}}}
\eqdef
\mathtt{app}(\sembrackets{\ea},\sembrackets{\eb})$
)
given enough auxiliary functions.

\section{Conclusion}

At which point in a classical realizability proof is the ``real work''
done? We have seen that the computational content of adequacy is
a normalization function, that can be annotated to reconstruct the
reduction sequence. Yet we have shown that there is very little leeway
in proofs of adequacy: their computational content is determined by
the \emph{types} of the adequacy lemma and of truth and falsity
witnesses.
We do not claim to have explored the entire design space, but would be
tempted to conjecture that there is a unique pure (that is, without effect) program modulo $\beta\eta$-equivalence inhabiting the dependent type of $\rea$.

Finally, we have seen that \emph{polarity} plays an
important role in adequacy programs -- even when they finally
correspond to well-known untyped reduction strategies such as
call-by-name or call-by-value. This is yet another argument to study
the design space of type-directed or, more generally,
polarity-directed reduction strategies.

\begin{acks}
  We were fortunate to discuss this work with Guillaume
  Munch-Maccagnoni, Jean-Philippe Bernardy, Hugo Herbelin and Marc
  Lasson, who were patient and pedagogical in their explanations. We
  furthermore thank Adrien Guatto, Guillaume Munch-Maccagnoni,
  {\'E}tienne Miquey, Philip Wadler and anonymous reviewers for their
  feedback on the article.
\end{acks}

\clearpage

\bibliography{rea}

\end{document}
